\documentclass[11pt, onecolumn]{article}
\usepackage[top=1in, bottom=1in, left=1.25in, right=1.25in]{geometry}
\usepackage{amsfonts}
\usepackage{amsmath}
\usepackage{graphicx}
\usepackage{color, soul}
\usepackage{algorithm,algorithmic}
\usepackage{bm}
\usepackage{booktabs}
\usepackage{flushend}

\usepackage[amsmath,thmmarks]{ntheorem}
\usepackage{theorem}

\newtheorem{pro}{Proposition}
\newtheorem{cor}{Corollary}
\newtheorem{lem}{Lemma}
\newtheorem{defi}{Definition}
\newtheorem{rem}{Remark}
\newtheorem{thm}{Theorem}

\theoremheaderfont{\sc}\theorembodyfont{\upshape}
\theoremstyle{nonumberplain}
\theoremseparator{}
\theoremsymbol{\rule{1ex}{1ex}}
\newtheorem{proof}{Proof}

\renewcommand{\arraystretch}{1.5}

\renewcommand{\arraystretch}{1.5}

\begin{document}
\title{Local-set-based Graph Signal Reconstruction}
\author{Xiaohan~Wang, Pengfei~Liu, and Yuantao~Gu%
\thanks{
The authors are with the Department of Electronic Engineering, Tsinghua University, Beijing 100084, CHINA. The
corresponding author of this paper is Yuantao Gu (gyt@tsinghua.edu.cn).
}}

\date{Submitted Sep. 5, 2014; Revised Jan. 18, 2015; Accepted Feb. 11, 2015.}

\maketitle

\begin{abstract}
Signal processing on graph is attracting more and more attentions. For a graph signal in the low-frequency subspace, the missing data associated with unsampled vertices can be reconstructed through the sampled data by exploiting the smoothness of the graph signal.
In this paper, the concept of local set is introduced and two local-set-based iterative methods are proposed to reconstruct bandlimited graph signal from sampled data.
In each iteration, one of the proposed methods reweights the sampled residuals for different vertices, while the other propagates the sampled residuals in their respective local sets.
These algorithms are built on frame theory and the concept of local sets, based on which several frames and contraction operators are proposed. We then prove that the reconstruction methods converge to the original signal under certain conditions and demonstrate the new methods lead to a significantly faster convergence compared with the baseline method.
Furthermore, the correspondence between graph signal sampling and time-domain irregular sampling is analyzed comprehensively, which may be helpful to future works on graph signals. Computer simulations are conducted. The experimental results demonstrate the effectiveness of the reconstruction methods in various sampling geometries, imprecise priori knowledge of cutoff frequency, and noisy scenarios.

\end{abstract}

\section{Introduction}
\subsection{Signal Processing on Graph}
In recent years, the increasing demands for signal and information processing in irregular domains have resulted in an emerging field of signal processing on graphs \cite{shuman_emerging_2013,sandryhaila_discrete_2013}. Bringing a new perspective for analyzing data associated with graphs, graph signal processing has found potential applications in sensor networks \cite{zhu_graph_2012}, image processing \cite{narang_graph_2012}, semi-supervised learning \cite{gadde_active_2014}, and recommendation systems \cite{narang_signal_2013}.

An undirected graph is denoted as $\mathcal{G}(\mathcal{V}, \mathcal{E})$, where $\mathcal{V}$ denotes a set of $N$ vertices and $\mathcal{E}$ denotes the edge set. If one real number is associated with each vertex, these numbers of all the vertices are collectively referred as a graph signal. A graph signal can also be regarded as a mapping $f: \mathcal{V}\rightarrow \mathbb{R}$.

There has been lots of research on graph signal related problems, including graph filtering \cite{zhang_graph_2008, chen_adaptive_2013}, graph wavelets \cite{Crovella_Graph_2003, Coifman_Diffusion_2006, hammond_wavelets_2011, narang_perfect_2012}, uncertainty principle \cite{agaskar_aspectral_2013}, multiresolution transforms \cite{shuman_aframework_2013, ekambaram_multiresolution_2013}, graph signal compression \cite{zhu_approximating_2012}, graph signal sampling \cite{narang_localized_2013, anis_towards_2014}, parametric dictionary learning \cite{thanou_parametric_2013}, graph topology learning \cite{dong_learning_2014}, and graph signal coarsening \cite{liu_coarsening_2014}.

\subsection{Motivation and Related Works}

Smooth signals or approximately smooth signals over graph are common in practical applications \cite{sandryhaila_discrete_2013,narang_signal_2013,chen_adaptive_2013,chen_signal_2014}, especially for those cases in which the graph topologies are constructed to enforce the smoothness property of signals \cite{dong_learning_2014}. Exploiting the smoothness of a graph signal, it may be reconstructed through its entries on only a part of the vertices, i.e. samples of the graph signal.

In this work, we develop efficient methods to solve the problem of reconstructing a bandlimited graph signal from known samples. The smooth signal is supposed to be within a low-frequency subspace. Two iterative methods are proposed to recover the missing entries from known sampled data.

There has been some theoretical analysis on the sampling and reconstruction of bandlimited graph signals \cite{pesenson_sampling_2008, pesenson_variational_2009,pesenson_sampling_2010,fuhr_poincare_2013}.
Some existing works focus on the theoretical conditions for the exact reconstruction of bandlimited signals. The relationships between the sampling sets of unique reconstruction and the cutoff frequency of bandlimited signal space are established for normalized Laplacian \cite{pesenson_sampling_2008} and unnormalized Laplacian \cite{pesenson_sampling_2010,fuhr_poincare_2013}, respectively.
Recently, a necessary and sufficient condition of exact reconstruction is established in \cite{anis_towards_2014}. In order to reconstruct bandlimited graph signals from sampled data, several methods have been proposed. In \cite{narang_signal_2013} a least square approach is proposed to solve this problem. Furthermore, an iterative reconstruction method is proposed and a tradeoff between smoothness and data-fitting is introduced for real world applications \cite{narang_localized_2013}.

The problem of signal reconstruction is closely related to the frame theory, which is also involved in other areas of graph signal processing, e.g., wavelet and vertex-frequency analysis on graphs \cite{hammond_wavelets_2011}. Based on windowed graph Fourier transform and vertex-frequency analysis, windowed graph Fourier frames are studied in \cite{shuman_vertex_2013}. A spectrum-adapted tight vertex-frequency frame is proposed in \cite{shuman_spectrum_2013} via translation on the graph. These works focus on vertex-frequency frames whose elements make up over-representation dictionaries, while in the reconstruction problem the frames are always composed by elements centering at the vertices in the sampling sets.

\subsection{Contributions}

In this paper, to improve the convergence rate of bandlimited graph signal reconstruction, iterative weighting reconstruction (IWR) and iterative propagating reconstruction (IPR) are proposed based a new concept of local set. As the foundation of reconstruction methods, several local-set-based frames and contraction operators are introduced. Both IWR and IPR are theoretically proved to uniquely reconstruct the original signal under certain conditions. Compared with existing methods, the condition of the proposed reconstruction methods is easy to determine by local parameters. The correspondence between graph signal sampling and time-domain irregular sampling is analyzed comprehensively, which will be helpful to future works on graph signals. Experiments show that IWR and IPR converge significantly faster than available methods. Besides, experiments on several topics including sampling geometry and robustness are conducted.

The rest of this paper is organized as follows. In Section \ref{sec2}, some preliminaries are introduced. In Section \ref{sec3}, some important definitions are introduced and some related frames based on local sets are proved.
In Section \ref{sec4}, two local-set-based reconstruction methods IWR and IPR are proposed and their convergence behavior is analyzed, respectively.
Section \ref{discusslocalset} gives more detailed analysis on local sets.
Section \ref{sec6} shows the relationship between graph signal sampling and time-domain irregular sampling and Section \ref{sec7} presents some numerical experiments.

\section{Preliminaries}\label{sec2}

\subsection{Graph Laplacian and Bandlimited Graph Signals}

The graph Laplacian is  extensively exploited in spectral graph theory \cite{chung_spectral_1997} and signal processing on graphs \cite{shuman_emerging_2013}.
For a undirected graph $\mathcal{G}(\mathcal{V}, \mathcal{E})$, its Laplacian is
$$\bf{L=D-A},$$
where $\bf{A}$ is the adjacency matrix of the graph and $\bf{D}$ is a diagonal degree matrix with the diagonal elements as the degrees of corresponding vertices.

The Laplacian is a real symmetric matrix, and all the eigenvalues are nonnegative.
Supposing $\{\lambda_k\}$ are the eigenvalues, and $\{{\bf u}_k\}$ are the corresponding eigenvectors,
the graph Fourier transform is defined as the expansion of a graph signal ${\bf f}$ in terms of $\{{\bf u}_k\}$, as
$$
\hat{f}(k)=\langle {\bf f}, {\bf u}_k\rangle=\sum_{i=1}^Nf(i)u_k(i),
$$
where $f(i)$ denotes the entry of ${\bf f}$ associated with vertex $i$.
Similar with classical Fourier analysis, eigenvalues $\{\lambda_k\}$ are regarded as frequencies of the graph, and $\hat{f}(k)$ is regarded as the frequency component corresponding to $\lambda_k$.
The frequency components associated with smaller eigenvalues can be called low-frequency part, and those associated with larger eigenvalues is the high-frequency part.

For a graph signal ${\mathbf f}\in\mathbb{R}^N$ on a graph $\mathcal{G}(\mathcal{V},\mathcal{E})$,
${\mathbf f}$ is called $\omega$-bandlimited if the spectral support of ${\mathbf f}$ is within $[0,\omega]$.
That is, the frequency components corresponding to eigenvalues larger than $\omega$ are all zero.
The subspace of $\omega$-bandlimited signals on graph $\mathcal{G}$ is a Hilbert space called Paley-Wiener space, denoted as $PW_{\omega}(\mathcal{G})$ \cite{pesenson_sampling_2008}.

In this paper, we consider the sampling and reconstruction of bandlimited signals on undirected and unweighted graphs. Suppose that for a bandlimited graph signal ${\mathbf f}\in PW_{\omega}(\mathcal{G})$, only $\{f(u)\}_{u\in \mathcal{S}}$ on the sampling set $\mathcal{S}\subseteq \mathcal{V}$ are known, the problem is to obtain the original signal ${\mathbf f}$ from the sampled data.

\subsection{Frame Theory and Signal Reconstruction}

The problem of signal sampling and reconstruction is closely related to frame theory.

\begin{defi}[frame and frame bound]
A family of elements $\{{\mathbf f}_i\}_{i\in \mathcal{I}}$ is a \emph{frame} for a Hilbert space $\mathcal{H}$, if there exist constants $0<A\le B$ such that
$$
A\|{\mathbf f}\|^2\le\sum_{i\in \mathcal{I}}|\langle {\mathbf f}, {\mathbf f}_i\rangle|^2\le B\|{\mathbf f}\|^2,\quad \forall {\mathbf f}\in \mathcal{H},
$$
where $A$ and $B$ are called \emph{frame bounds}.
\end{defi}

\begin{defi}[frame operator]\label{def:frameoperator}
For a frame $\{{\mathbf f}_i\}_{i\in \mathcal{I}}$, \emph{frame operator} ${\bf S}: \mathcal{H}\rightarrow \mathcal{H}$ is defined as
$$
{\mathbf{Sf}}=\sum_{i\in \mathcal{I}}\langle {\mathbf f}, {\mathbf f}_i\rangle {\mathbf f}_i.
$$
\end{defi}

One may readily read that $A{\mathbf I}\preceq {\mathbf S}\preceq B{\mathbf I}$ for $\mathcal{H}$, where ${\mathbf I}$ denotes the identity operator and $A{\mathbf I}\preceq {\mathbf S}$ means that ${\mathbf S}-A{\mathbf I}$ is positive semidefinite. Consequently, ${\mathbf S}$ is always invertible and its inverse could be expanded into series in some special cases. For instance, one has
$$
{\mathbf f}={\bf S}^{-1}{\bf S}{\mathbf f}=\mu \sum_{j=0}^{\infty}({\bf I}-\mu {\bf S})^j{\bf S}{\mathbf f},
$$
where $\mu$ is a scalar satisfying $\|{\bf I}-\mu {\bf S}\|<1$. This inspires that $\bf f$ could be iteratively reconstructed from any initial point ${\bf f}^{(0)}$ by
\begin{align}
{\mathbf f}^{(k+1)}&=\mu {\bf S}{\mathbf f}+({\bf I}-\mu {\bf S}){\mathbf f}^{(k)}\nonumber\\
&={\mathbf f}^{(k)}+\mu {\bf S}({\mathbf f}-{\mathbf f}^{(k)}),\label{iteration}
\end{align}
with the error bound satisfying
$$
\|{\mathbf f}^{(k)}-{\mathbf f}\|\le\|{\bf I}-\mu {\bf S}\|^{k}\|{\mathbf f}^{(0)}-{\mathbf f}\|.
$$
Obviously, recursion \eqref{iteration} cannot be entitled \emph{reconstruction} because the original signal to be recovered is involved in the iteration. However, it provides a prototype for practical methods, which will be discussed in section \ref{sec4}.

The parameter $\mu$, which could be deemed as a step-size, determines the convergence rate.
If one chooses $\mu=1/B$, then $\|{\bf I}-\mu {\bf S}\|\le1-A/B<1$, and the error bound of iteration (\ref{iteration}) will shrink with the exponential of $(1-A/B)$.
A better choice is $\mu=2/(A+B)$, then $\|{\bf I}-\mu {\bf S}\|\le(B-A)/(B+A)$, which leads to a faster convergence rate \cite{Christensen_an_2002}.

\subsection{Bandlimited Graph Signal Reconstruction}
\label{Previous}

There are many useful theoretical results
\footnote{
It is necessary to notice that some of the theoretical results are based on normalized Laplacian. However, similar results can be easily obtained for Laplacian, which is mainly used in this work.
}
on the problem of bandlimited graph signal sampling and reconstruction. A concept of \emph{uniqueness set} is firstly introduced in \cite{pesenson_sampling_2008}.

\begin{defi}[uniqueness set]
\cite{pesenson_sampling_2008} A set of vertices $\mathcal{S}\subseteq \mathcal{V}(\mathcal{G})$ is a \emph{uniqueness set} for space $PW_{\omega}(\mathcal{G})$ if it holds for all ${\mathbf f}, {\mathbf g} \in PW_{\omega}(\mathcal{G})$ that $f(u)=g(u), \forall u\in\mathcal{S}$ implies ${\mathbf f}={\mathbf g}$.
\end{defi}

According to this definition, any ${\mathbf f}\in PW_{\omega}(\mathcal{G})$ could be uniquely determined by its entries on a uniqueness set $\mathcal{S}$. As a consequence, ${\bf f}$ may be exactly recovered if the sampling set is a uniqueness set. Readers are suggested to refer to \cite{pesenson_sampling_2008}, \cite{narang_signal_2013}, and \cite{anis_towards_2014} for more details on uniqueness set.

The following theorem demonstrates that a set of graph signals related to a uniqueness set becomes a frame for $PW_{\omega}(\mathcal{G})$, which is a quite important foundation of our work.

\begin{thm}\label{thm3}
\cite{pesenson_sampling_2008} If the sampling set $\mathcal{S}$ is a uniqueness set for $PW_{\omega}(\mathcal{G})$, then $\{\mathcal{P}_{\omega}(\bm{\delta}_u)\}_{u\in \mathcal{S}}$ is a frame for $PW_{\omega}(\mathcal{G})$, where $\mathcal{P}_{\omega}(\cdot)$ is the projection operator onto $PW_{\omega}(\mathcal{G})$, and $\bm{\delta}_u$ is a $\delta$-function whose entries satisfying
$$
\delta_u(v)=
\begin{cases}
1, & v=u; \\
0, & v\neq u.
\end{cases}
$$
\end{thm}

A method called iterative least square reconstruction (ILSR) is proposed to reconstruct bandlimited graph signals in \cite{narang_localized_2013} as the following theorem.

\begin{thm}\label{thm:ILSR}
\cite{narang_localized_2013} If the sampling set $\mathcal{S}$ is a uniqueness set for $PW_{\omega}(\mathcal{G})$, then the original signal ${\bf f}$ can be reconstructed using the sampled data $\{f(u)\}_{u\in\mathcal{S}}$
by ILSR method,
\begin{equation}\label{ilsr}
{\bf f}^{(k+1)}={\mathcal P}_{\omega}({\bf f}^{(k)}+{\bf J}^{\rm T}{\bf J}({\bf f}_{\text{du}}-{\bf f}^{(k)})),
\end{equation}
where ${\bf J}$ denotes the downsampling operator and ${\bf f}_{\text{du}}$ is the downsampled signal.
\end{thm}

ILSR is derived from the method of projection onto convex sets (POCS). Its convergence is proved using the fixed point theorem of contraction mapping.

\section{Local-Set-Based Frame and Contraction}
\label{sec3}

In this section, the concept of \emph{local set} is firstly proposed. Based on local sets, we define an operator named \emph{local propagation} and prove its contraction. Then several local-set-based frames are introduced, as the theoretical foundation of the proposed methods in next section.

\subsection{Local Sets}

\begin{defi}[local sets]
For a sampling set $\mathcal{S}$ on graph $\mathcal{G}(\mathcal{V},\mathcal{E})$, assume that $\mathcal{V}$ is divided into disjoint local sets $\{\mathcal{N}(u)\}_{u\in \mathcal{S}}$ associated with the sampled vertices. For each $u\in\mathcal{S}$, denote the subgraph of $\mathcal{G}$ restricted to $\mathcal{N}(u)$ by $\mathcal{G}_{\mathcal{N}(u)}$, which is composed of vertices in $\mathcal{N}(u)$ and edges between them in $\mathcal{E}$. For each $u\in\mathcal{S}$, its local set satisfies $\mathcal{N}(u)\ni u$, and the subgraph $\mathcal{G}_{\mathcal{N}(u)}$ is connected.  Besides, $\{\mathcal{N}(u)\}_{u\in \mathcal{S}}$ should satisfy
$$
\mathcal{N}(u)\cap \mathcal{N}(v)=\emptyset, \quad \forall u, v\in\mathcal{S} \,\textrm{and}\, u\neq v,
$$
and
$$
\bigcup_{u\in \mathcal{S}} \mathcal{N}(u)=\mathcal{V}.
$$

We further define
$$
	N_{\rm max}=\max\limits_{u\in\mathcal{S}}|\mathcal{N}(u)|,
$$
to denote the maximal size of local sets, where $|\cdot|$ denotes cardinality.
\end{defi}

For a given sampling set, there may exist various divisions of local sets.
We will see in the next section that different divisions may lead to different theoretical bounds in recovering bandlimited signals.
To describe the property of local sets, two measures are proposed in Definition \ref{defi1} and Definition \ref{radius}, which are useful in the following analysis.

\begin{defi}[maximal multiple number]\label{defi1}
Denote
$$
\mathcal{T}(u)={\rm SPT}(\mathcal{G}_{\mathcal{N}(u)})
$$
as the shortest-path tree of $\mathcal{G}_{\mathcal{N}(u)}$ rooted at $u$.
For $v$ connected to $u$ in $\mathcal{T}(u)$, $\mathcal{T}_u(v)$ is the subtree which $v$ belongs to when $u$ and its associated edges are removed from $\mathcal{T}(u)$.
The \emph{maximal multiple number} of $\mathcal{N}(u)$ is defined as
$$
K(u)=\max_{(u,v)\in \mathcal{E}(\mathcal{T}(u))}|\mathcal{T}_u(v)|,
$$
where $\mathcal{E}(\mathcal{T}(u))$ is the edge set of graph $\mathcal{T}(u)$.
\end{defi}

\begin{figure}[t]
\begin{center}
\includegraphics[width=10cm]{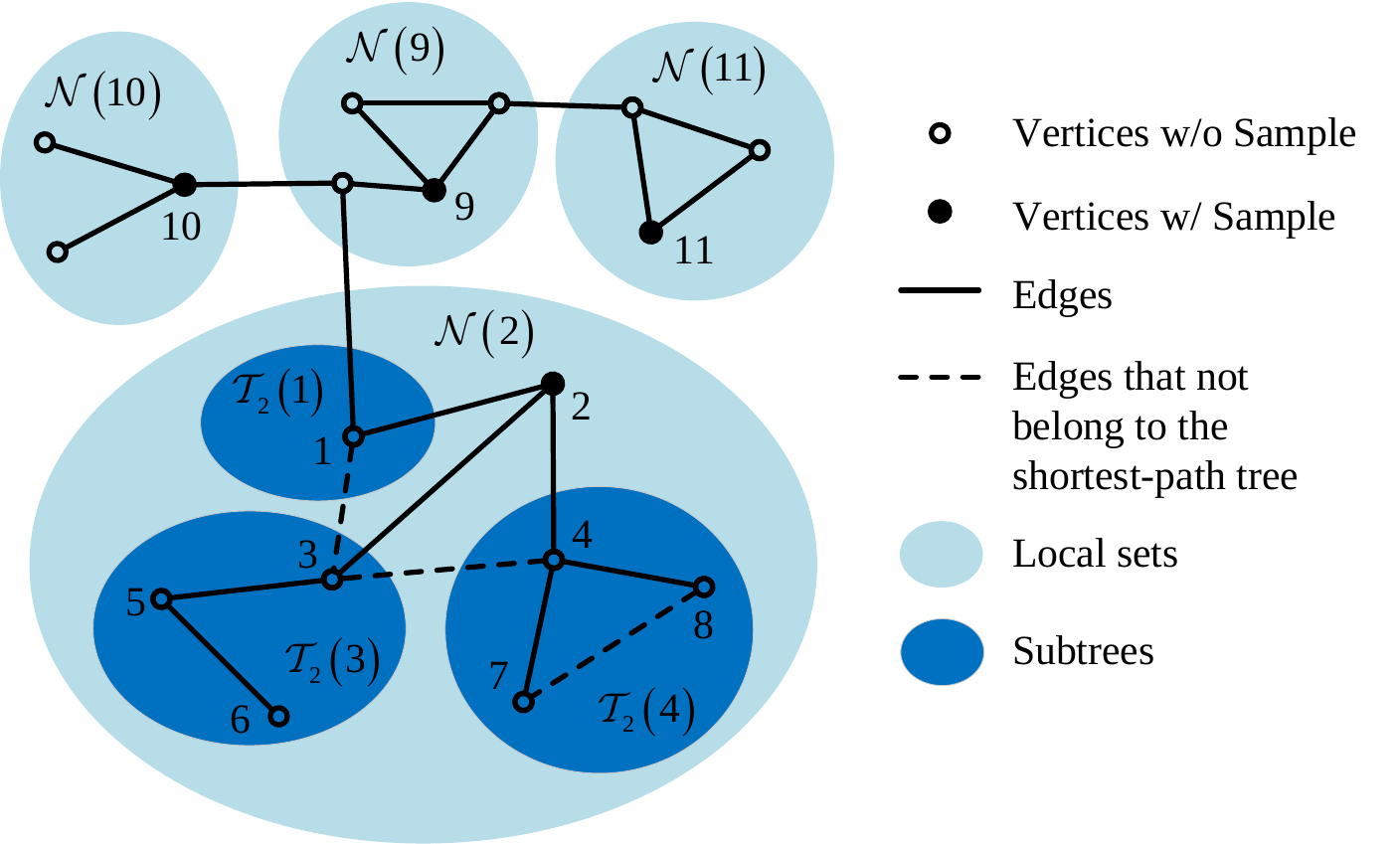}
\caption{
Illustrations of an example graph $\mathcal{G}$, sampling set ${\mathcal S}=\{2, 9, 10, 11\}$, and one of the divisions of local sets $\{\mathcal{N}(2), \mathcal{N}(9), \mathcal{N}(10), \mathcal{N}(11)\}$. The shortest-path tree of $\mathcal{N}(2)$ and its subtrees of $\mathcal{T}_2(1), \mathcal{T}_2(3), \mathcal{T}_2(4)$ are highlighted to give more details. In this case, $|\mathcal{N}(2)|=8$ and the subtrees have $1$, $3$, and $3$ vertices, respectively. Therefore, $K(2)=3$ and it is easy to check that $\tilde{K}(2)=|\mathcal{N}(2)|-d_{\mathcal{N}(2)}(2)=8-3=5$.
}
\label{Tu}
\end{center}
\end{figure}

\begin{rem}
By the definition of $K(u)$, it is ready to check that
\begin{equation}\label{ku}
K(u)\le |\mathcal{N}(u)|-d_{\mathcal{N}(u)}(u)\le |\mathcal{N}(u)|-1,
\end{equation}
where $d_{\mathcal{N}(u)}(u)$ is the degree of $u$ in the subgraph $\mathcal{G}_{\mathcal{N}(u)}$.
For simplicity, one may introduce an approximation for easy calculation of $K(u)$ by
\begin{equation}\label{tildeku}
\tilde{K}(u)=|\mathcal{N}(u)|-d_{\mathcal{N}(u)}(u).
\end{equation}
\end{rem}

The definitions above are intuitively illustrated in Fig. \ref{Tu}.

\begin{defi}[radius]\label{radius}
The \emph{radius} of $\mathcal{N}(u)$ is the maximal distance from $u$ to any other vertex in $\mathcal{G}_{\mathcal{N}(u)}$, which is denoted as
$$
R(u)=\max_{v\in \mathcal{N}(u)}\text{dist}(v,u).
$$
\end{defi}

According to the definitions, one may see that the two measures are local and only determined by the subgraph of local sets.
The two local measures are helpful in establishing the conditions of some important results in this paper, which will be shown in the following subsections.

\subsection{Local Propagation and Contraction}

Utilizing the introduced maximal multiple number and radius, we further propose an operation to propagate energy to one's local set.

\begin{defi}[local propagation]\label{limitedpropagation}
For a given sampling set $\mathcal{S}$ and associated local sets $\{\mathcal{N}(u)\}_{u\in\mathcal{S}}$ on a graph $\mathcal{G}(\mathcal{V},\mathcal{E})$, the local propagation ${\bf G}$ is defined by
\begin{align}
{\bf G}{\mathbf f}&=\mathcal{P}_{\omega}\left(\sum_{u\in \mathcal{S}}f(u)\bm{\delta}_{\mathcal{N}(u)}\right)\label{eq:limitedpropagation1}\\
&= \sum_{u\in \mathcal{S}}f(u)\mathcal{P}_{\omega}\!\left(\bm{\delta}_{\mathcal{N}(u)}\right),\label{eq:limitedpropagation2}
\end{align}
where $\bm{\delta}_{\mathcal{N}(u)}$ denotes the $\delta$-function of set $\mathcal{N}(u)$ with entries
$$
\delta_{\mathcal{N}(u)}(v)=
\begin{cases}
1, & v\in \mathcal{N}(u);\\
0, & v\notin \mathcal{N}(u).
\end{cases}
$$
\end{defi}

As its name shows, operation $\bf G$ first propagates the energy locally and evenly to the local set that each sampled vertex belongs to, and then projects the new signal to be $\omega$-bandlimited, please refer to \eqref{eq:limitedpropagation1}. These two steps could be merged into one, by a bandlimited local propagation of $\mathcal{P}_{\omega}\left(\bm{\delta}_{\mathcal{N}(u)}\right)$, please refer to \eqref{eq:limitedpropagation2}. Local propagation, which provides a fast solution to adequately fill all unknown entries by sampled data, makes the proposed local-set-based reconstruction feasible.

As an important theoretical foundation, the following lemma gives the condition that $({\bf I-G})$ is a contraction mapping.

\begin{lem}\label{lemma1}
For a given set $\mathcal{S}$ and associated local sets $\{\mathcal{N}(u)\}_{u\in\mathcal{S}}$ on a graph $\mathcal{G}(\mathcal{V},\mathcal{E})$, $\forall \omega<1/Q_{\rm max}^2$, the operator $({\bf I-G})$ is a contraction mapping for $PW_{\omega}(\mathcal{G})$, where
\begin{equation}\label{eq:defineQmax}
	Q_{\text{max}} =\max_{u\in \mathcal{S}}\sqrt{K(u)R(u)}.
\end{equation}
\end{lem}

\begin{proof}
The proof is postponed to \ref{proof1}.
\end{proof}

For a given sampling set, there may exist various divisions of local sets. We will see in the next section that different divisions may lead to different theoretical bounds in recovering bandlimited signals.

\subsection{Weighted Frame}\label{weightedframe}

Based on the definition of local set, we could prove that the weighted lowpass $\delta$-function set is a frame for $PW_{\omega}(\mathcal{G})$ and estimate its bounds.

\begin{lem}\label{lem0}
For a given sampling set $\mathcal{S}$ and associated local sets $\{\mathcal{N}(u)\}_{u\in\mathcal{S}}$ on a graph $\mathcal{G}(\mathcal{V},\mathcal{E})$, $\forall \omega < 1/Q_{\rm max}^2$, $\{\mathcal{P}_{\omega}(\bm{\delta}_u)\}_{u\in \mathcal{S}}$ is a frame for $PW_{\omega}(\mathcal{G})$ with bounds $(1-\gamma)^2/N_{\text{max}}$ and $1$, where $Q_{\text{max}}$ is defined in \eqref{eq:defineQmax} and
\begin{equation}\label{eq:definegamma}
	\gamma =Q_{\rm max}\sqrt{\omega}.
\end{equation}
\end{lem}

\begin{proof}
The proof is postponed to \ref{proof3}.
\end{proof}

%\begin{rem}
%By applying Lemma \ref{lem0}, we could also derive ILSR \cite{narang_localized_2013} based on frame theory, which will be specified in section \ref{sec4}.
%\end{rem}

One may notice that Theorem 3.2 of \cite{pesenson_sampling_2010} also implies that $\{\mathcal{P}_{\omega}(\bm{\delta}_u)\}_{u\in \mathcal{S}}$ is a frame for $PW_{\omega}(\mathcal{G})$. However, the assumptions and approaches in this work are quite different from those in the above reference. Furthermore, base on the proposed local sets, we study the relation between sampling set and whole vertices, and clarify the frame bound exactly.

Beyond Lemma \ref{lem0}, we further explore the weighted lowpass $\delta$-functions is also a frame for $PW_{\omega}(\mathcal{G})$ by appropriate weights.

\begin{lem}\label{lemma3}
For a given sampling set $\mathcal{S}$ and associated local sets $\{\mathcal{N}(u)\}_{u\in\mathcal{S}}$ on a graph $\mathcal{G}(\mathcal{V},\mathcal{E})$, $\forall \omega < 1/Q_{\rm max}^2$, $\{\sqrt{|\mathcal{N}(u)|}\mathcal{P}_{\omega}(\bm{\delta}_u)\}_{u\in \mathcal{S}}$ is a frame for $PW_{\omega}(\mathcal{G})$ with bounds $\left(1-\gamma\right)^2$ and $\left(1+\gamma\right)^2$, where $Q_{\rm max}$ and $\gamma$ are defined in \eqref{eq:defineQmax} and \eqref{eq:definegamma}, respectively.
\end{lem}

\begin{proof}
The proof is postponed to \ref{proof4}.
\end{proof}

Bandlimited graph signals can be iteratively reconstructed using a frame for $PW_{\omega}(\mathcal{G})$, but the frame bounds play critical roles on the convergence rate. By given appropriate weights to the elements in a frame, a new frame is obtained with a sharper bounds estimation, which may lead to a faster convergence. The related algorithms will be proposed in section \ref{sec4}.

\subsection{Local Set Frame}

To end up this section, we present a general theoretical result which may inspire further study on frame-theory-based graph signal processing.

\begin{pro}\label{pro3}
For a given sampling set $\mathcal{S}$ and associated local sets $\{\mathcal{N}(u)\}_{u\in\mathcal{S}}$ on a graph $\mathcal{G}(\mathcal{V},\mathcal{E})$, $\forall \omega < 1/Q_{\rm max}^2$, $\{\mathcal{P}_{\omega}(\bm{\delta}_{\mathcal{N}(u)})\}_{u\in \mathcal{S}}$ is a frame for $PW_{\omega}(\mathcal{G})$ with bounds $(1-\gamma)^2$ and $N_{\max}$, where $Q_{\rm max}$ and $\gamma$ are defined in \eqref{eq:defineQmax} and \eqref{eq:definegamma}, respectively.
\end{pro}

\begin{proof}
The proof is postponed to \ref{proof3}.
\end{proof}

In fact local propagation is not a standard frame operator, because two signal sets $\{\mathcal{P}_{\omega}(\bm{\delta}_u)\}_{u\in \mathcal{S}}$ and $\{\mathcal{P}_{\omega}(\bm{\delta}_{\mathcal{N}(u)})\}_{u\in \mathcal{S}}$ are involved. However, under the same condition with the contraction of operator $({\bf I-G})$, both sets can be proved to be frame, and either of them can be used to reconstruct the original signal by the corresponding frame operator. All frames discussed in this section are listed in Table \ref{tab:frames}.

\begin{table}[t]
\caption{The frames in space $PW_{\omega}(\mathcal{G}), \forall \omega < 1/Q_{\max}^2$, and their bounds.}\label{tab:frames}
\begin{center}
\begin{tabular}{ccc}
\toprule[1pt]
Frame & Lower bound & Upper bound\\ \hline
$\{\mathcal{P}_{\omega}({\bm \delta}_u)\}_{u\in\mathcal{S}}$ &
$(1-\gamma)^2/N_{\max}$ & 1\\
$\{\sqrt{|\mathcal{N}(u)|}\mathcal{P}_{\omega}({\bm \delta}_u)\}_{u\in\mathcal{S}}$ & $(1-\gamma)^2$ & $(1+\gamma)^2$\\
$\{\mathcal{P}_{\omega}({\bm \delta}_{\mathcal{N}(u)})\}_{u\in\mathcal{S}}$ & $(1-\gamma)^2$ & $N_{\max}$\\
\bottomrule[1pt]
\end{tabular}
\end{center}
\end{table}

\section{Iterative Reconstruction Algorithms}
\label{sec4}

In this section, ILSR is represented in the frame-based framework. Then two novel methods IWR and IPR are proposed with theoretical analysis of convergence.
Discussions on the three methods are also included in this section.

\subsection{Iterative Least Square Reconstruction}

In this subsection, we will represent ILSR, which is proposed in \cite{narang_localized_2013} in the form of (\ref{ilsr}), into frame-based framework.

According to Definition \ref{def:frameoperator} and Lemma \ref{lem0}, frame operator associated with frame $\{\mathcal{P}_{\omega}(\bm{\delta}_u)\}_{u\in \mathcal{S}}$ is
\begin{equation}
{\bf Sf}=\sum_{u\in\mathcal{S}}\langle {\bf f}, \mathcal{P}_{\omega}(\bm{\delta}_u)\rangle \mathcal{P}_{\omega}(\bm{\delta}_u).\label{eq:frameonuniqueset}
\end{equation}
For ${\bf f}\in PW_{\omega}(\mathcal{G})$, one has $\mathcal{P}_{\omega}({\bf f})={\bf f}$ and yields
$$
\langle {\bf f}, \mathcal{P}_{\omega}(\bm{\delta}_u)\rangle=\langle \mathcal{P}_{\omega}({\bf f}),\bm{\delta}_u\rangle=\langle {\bf f},\bm{\delta}_u\rangle=f(u).
$$
Consequently, frame operatior \eqref{eq:frameonuniqueset} is reduced to
\begin{equation}
{\bf Sf}=\sum_{u\in\mathcal{S}}f(u) \mathcal{P}_{\omega}(\bm{\delta}_u).\label{eq:frameonuniquesetreduced}
\end{equation}
Utilizing \eqref{eq:frameonuniquesetreduced} in \eqref{iteration}, one may read that the original signal, whose unsampled values are never needed in the iterative reconstruction, could be exactly recovered from its entries on a uniqueness set. The reformulated ILSR method is displayed in Table \ref{algILSR}.

\begin{table}[t]
\renewcommand{\arraystretch}{1.2}
\caption{Iterative Least Square Reconstruction.}\label{algILSR}
\begin{center}
\begin{tabular}{l}
\toprule[1pt]
{\bf Input:} \hspace{0.5em} Graph $\mathcal{G}$, cutoff frequency $\omega$, sampling set $\mathcal{S}$,\\
\hspace{3.5em} sampled data $\{f(u)\}_{u\in\mathcal{S}}$;\\
{\bf Output:} \hspace{0.5em} Interpolated signal ${\bf f}^{(k)}$;\\
\hline
{\bf Initialization:}\\
\hspace{1.3em} $\displaystyle {\bf f}^{(0)}=\mathcal{P}_{\omega}\left(\sum_{u\in \mathcal{S}}f(u)\bm{\delta}_{u}\right);$\\
{\bf Loop:}\\
\hspace{1.3em} $\displaystyle {\mathbf f}^{(k+1)}={\mathbf f}^{(k)}+\mathcal{P}_{\omega}\left(\sum_{u\in \mathcal{S}}(f(u)-f^{(k)}(u))\bm{\delta}_{u}\right)$;\\
{\bf Until:}\hspace{0.5em} The stop condition is satisfied.\\
\bottomrule[1pt]
\end{tabular}
\end{center}
\end{table}

\subsection{Iterative Weighting Reconstruction}

Using the weighted frame, an algorithm named iterative weighting reconstruction (IWR) is proposed in Proposition \ref{pro2} and its convergence is proved.

\begin{pro}\label{pro2}
For a given sampling set $\mathcal S$ and associated local sets $\{\mathcal{N}(u)\}_{u\in\mathcal{S}}$ on a graph $\mathcal{G}(\mathcal{V},\mathcal{E})$, $\forall {\mathbf f}\in PW_{\omega}(\mathcal{G})$, where $\omega<1/Q_{\rm max}^2$, ${\bf f}$ can be reconstructed by the sampled data $\{f(u)\}_{u\in \mathcal{S}}$ through the IWR method in Table \ref{algIWR}, with the error bound satisfying
$$
\|{\mathbf f}^{(k)}-{\mathbf f}\|\le \left(\frac{2\gamma}{1+\gamma^2}\right)^{k}\|{\mathbf f}^{(0)}-{\mathbf f}\|,
$$
where $Q_{\rm max}$ and $\gamma$ are defined in \eqref{eq:defineQmax} and \eqref{eq:definegamma}, respectively.
\end{pro}

\begin{table}[t]
\renewcommand{\arraystretch}{1.2}
\caption{Iterative Weighting Reconstruction.}\label{algIWR}
\begin{center}
\begin{tabular}{l}
\toprule[1pt]
{\bf Input:} \hspace{0.5em} Graph $\mathcal{G}$, cutoff frequency $\omega$, sampling set $\mathcal{S}$,\\
\hspace{3.5em} neighbor sets $\{\mathcal{N}(u)\}_{u\in\mathcal{S}}$, sampled data $\{f(u)\}_{u\in\mathcal{S}}$;\\
{\bf Output:} \hspace{0.5em} Interpolated signal ${\bf f}^{(k)}$;\\
\hline
{\bf Initialization:}\\
\hspace{1.3em} $\displaystyle {\mathbf f}^{(0)}=\frac{1}{1+\gamma^2}\mathcal{P}_{\omega}\left(\sum_{u\in \mathcal{S}}|\mathcal{N}(u)|f(u)\bm{\delta}_{u}\right);$\\
{\bf Loop:}\\
\hspace{1.3em} $\displaystyle {\mathbf f}^{(k+1)}={\mathbf f}^{(k)}+\frac{1}{1+\gamma^2}\mathcal{P}_{\omega}\left(\sum_{u\in \mathcal{S}}|\mathcal{N}(u)|(f(u)-f^{(k)}(u))\bm{\delta}_{u}\right)$;\\
{\bf Until:}\hspace{0.5em} The stop condition is satisfied.\\
\bottomrule[1pt]
\end{tabular}
\end{center}
\end{table}

\begin{proof}
The proof is postponed to \ref{proof5}.
\end{proof}

The idea of IWR is to attach different weights to sampled vertices.
The weights for vertex $u$ is larger if its local set $\mathcal{N}(u)$ has more vertices, in other words, the vertex $u$ is more isolated or
the region around $u$ has a lower sampling density.
On the contrary, if the sampled vertices in a region are very dense, less importance is allocated to them.

\begin{cor}\label{pro4}
In Lemma \ref{lemma3} and Proposition \ref{pro2}, $Q_{\text{max}}$ can be replaced by $\tilde{Q}_{\text{max}}$, which is defined as
$$
\tilde{Q}_{\text{max}}=\max_{u\in \mathcal{S}}\sqrt{\tilde{K}(u)R(u)}.
$$
\end{cor}

According to (\ref{ku}), for any $u\in \mathcal{S}$ we have $\tilde{K}(u)\ge K(u)$, and then $\tilde{Q}_{\text{max}}\ge Q_{\text{max}}$.
In fact, $K(u)$  is not easy to obtain for each given subgraph $\mathcal{G}_{\mathcal{N}(u)}$.
However, $\tilde{K}(u)$ is convenient to get and $\tilde{Q}_{\text{max}}$
is a practical choice, even though the bound is not as accurate.

Proposition \ref{pro2} is a natural generalization of Theorem \ref{thm:ILSR} after ILSR is rebuilt on frame theory and
$\{\sqrt{|\mathcal{N}(u)|}\mathcal{P}_{\omega}(\bm{\delta}_{u})\}_{u\in \mathcal{S}}$ is proved to be a frame in Lemma \ref{lemma3}.

\subsection{Iterative Propagating Reconstruction}
Iterative propagating reconstruction (IPR) is proposed as the result of the contraction of the local propagation operator.

\begin{pro}\label{cor1}
For a given sampling set $\mathcal S$ and associated local sets $\{\mathcal{N}(u)\}_{u\in\mathcal{S}}$ on a graph $\mathcal{G}(\mathcal{V},\mathcal{E})$, $\forall {\mathbf f}\in PW_{\omega}(\mathcal{G})$, where $\omega<1/Q_{\rm max}^2$, ${\mathbf f}$ can always be reconstructed by its samples $\{f(u)\}_{u\in \mathcal{S}}$ through the IPR method in Table \ref{algIPR},
with the error bound satisfying
$$
\|{\mathbf f}^{(k)}-{\mathbf f}\|\le \gamma^{k}\|{\mathbf f}^{(0)}-{\bf f}\|,
$$
where $Q_{\rm max}$ and $\gamma$ are defined in \eqref{eq:defineQmax} and \eqref{eq:definegamma}, respectively.
\end{pro}

\begin{table}[t]
\renewcommand{\arraystretch}{1.2}
\caption{Iterative Propagating Reconstruction.}\label{algIPR}
\begin{center}
\begin{tabular}{l}
\toprule[1pt]
{\bf Input:} \hspace{0.5em} Graph $\mathcal{G}$, cutoff frequency $\omega$, sampling set $\mathcal{S}$,\\
\hspace{3.5em} neighbor sets $\{\mathcal{N}(u)\}_{u\in\mathcal{S}}$, sampled data $\{f(u)\}_{u\in\mathcal{S}}$;\\
{\bf Output:} \hspace{0.5em} Interpolated signal ${\bf f}^{(k)}$;\\
\hline
{\bf Initialization:}\\
\hspace{1.3em} $\displaystyle {\mathbf f}^{(0)}=\mathcal{P}_{\omega}\left(\sum_{u\in \mathcal{S}}f(u)\bm{\delta}_{\mathcal{N}(u)}\right);$\\
{\bf Loop:}\\
\hspace{1.3em} $\displaystyle {\mathbf f}^{(k+1)}={\mathbf f}^{(k)}+\mathcal{P}_{\omega}\left(\sum_{u\in \mathcal{S}}(f(u)-f^{(k)}(u))\bm{\delta}_{\mathcal{N}(u)}\right)$;\\
{\bf Until:}\hspace{0.5em} The stop condition is satisfied.\\
\bottomrule[1pt]
\end{tabular}
\end{center}
\end{table}

\begin{proof}
The proof is postponed to \ref{proof2}.
\end{proof}

Other than basic frame operators in ILSR and IWR, IPR is based on the contraction of the local propagation operator, in which two frames are involved.
Strictly speaking, IPR is not a frame-based method.
IPR is closed related to local sets, which is one of the main contributions of this work.

\begin{rem}
Similar to Proposition \ref{pro2}, $Q_{\text{max}}$ can also be replaced by $\tilde{Q}_{\text{max}}$ in Lemma \ref{lemma1} and Proposition \ref{cor1}, which is more practical to obtain.
\end{rem}

Since $\frac{2\gamma}{1+\gamma^2}>\gamma$ when $0<\gamma<1$, the theoretical guarantee of IPR decays faster than that of IWR. When $\gamma$ approaches to $1$, the two theoretical guarantees are close to each other.

\subsection{Intuitive Explanation of Three Algorithms}\label{intuitive}

\begin{figure*}
\begin{center}
\includegraphics[width=\textwidth]{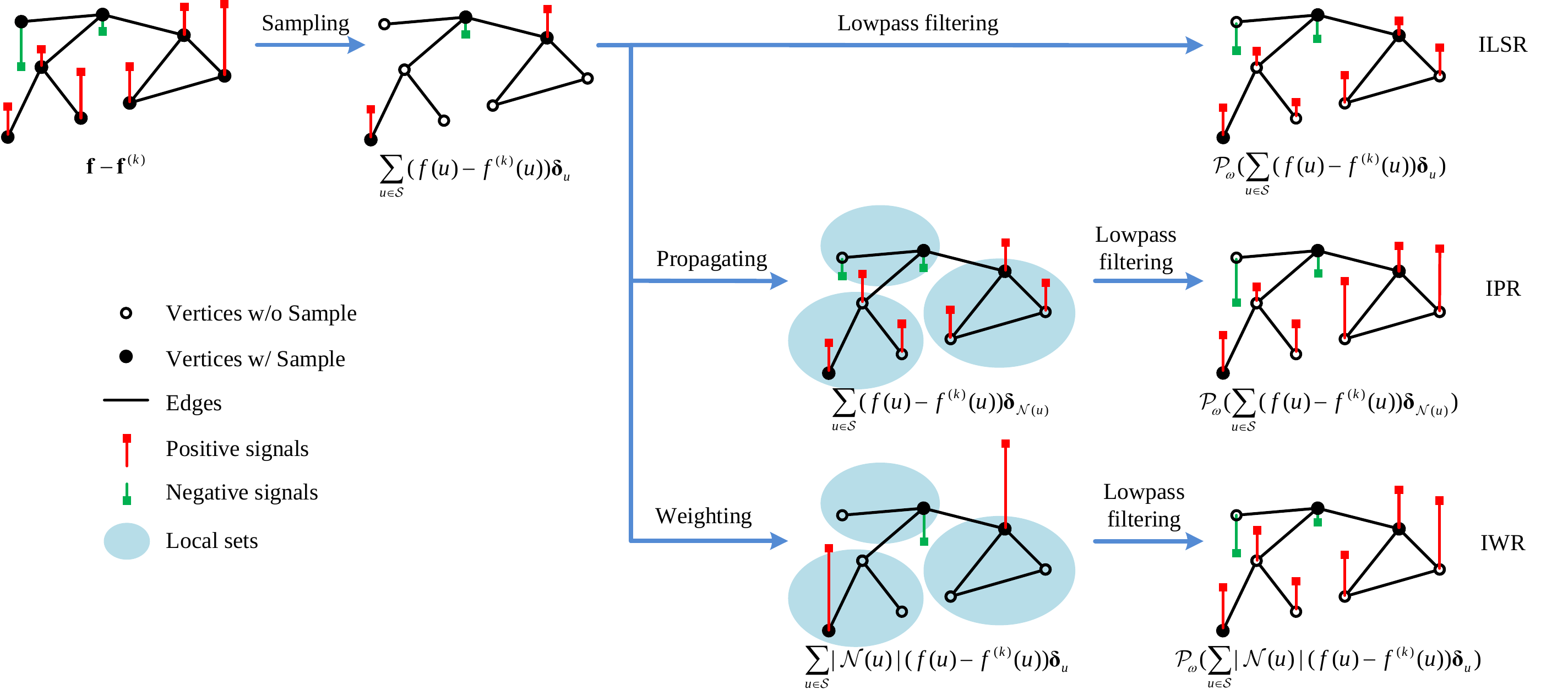}
\caption{Illustration of the iterations of the three algorithms. Compared with ILSR, IPR needs to assign the residuals on the sampled vertices to the vertices in the local set,
while IWR needs to multiply the residuals on the sampled vertices by different weights.}
\label{algorithm}
\end{center}
\end{figure*}

As illustrated in Fig. \ref{algorithm}, the differences among ILSR, IWR, and IPR lie in the way of their dealing with the residuals at the sampled vertices. In each iteration the sampled residual $\sum_{u\in\mathcal{S}}(f(u)-f^{(k)}(u))\bm{\delta}_u$ is directly projected onto the $\omega$-bandlimited space $PW_{\omega}(\mathcal{G})$ in ILSR. In IWR, the sampled residuals are multiplied by weights $|\mathcal{N}(u)|$ and then projected onto the low-frequency space. For IPR, the sampled residuals are copied and assigned to the vertices in the corresponding local sets and then the projection procedure is conducted.

Because of the weighting or propagating procedure, for each step the increment of IWR or IPR is larger than that of ILSR. It may intuitively explain why the proposed two algorithms both converge faster than ILSR.
Besides, it is easy to see from Fig. \ref{algorithm} that the graph signal composed of the propagated residuals seems closer to a low-frequency signal than the weighted residual, which means that the increment of IPR remains more than that of IWR after the projection to the low-frequency subspace. It may explain why IPR converges even faster than IWR.

\subsection{Discussions}
According to Proposition \ref{pro2} and \ref{cor1}, the estimated convergence bounds of IWR and IPR are related to the cutoff frequency and the topology of local sets.
Adequately estimating the cutoff frequency of the raw signal may accelerate the convergence of reconstruction.
For given $\mathcal S$ and $\{\mathcal{N}(u)\}_{u\in\mathcal{S}}$, the maximal multiple number and radius are to be calculated. Consequently, $Q_{\rm max}$ is determined. Therefore, a smaller known $\omega$ leads to a smaller $\gamma$, then sharper error bounds of convergence are obtained for both IWR and IPR, which may lead to a faster convergence.
The choice of local sets also affects the convergence performance, which will be discussed in the following section.

In the local-set-based methods IWR and IPR, the theoretical maximal cutoff frequency, below which the raw signal could be recovered, is much easier to obtain. For given sampling set $\mathcal S$ and associated local sets $\{\mathcal{N}(u)\}_{u\in\mathcal{S}}$, the maximal multiple number and radius can be obtained locally, then $Q_{\rm max}$ and the cutoff frequency are easy to be determined. It is quite different from that of uniqueness-set-based ILSR, where the sampling set are required to be a uniqueness set of specific cutoff frequency. The conditions for uniqueness set are determined by global measures such as eigenvalues, as shown in \cite{pesenson_sampling_2008} and \cite{narang_signal_2013}, which is rather difficult to obtain in large-scale problems.
However, according to Lemma \ref{lem0}, the condition of ILSR can also be modified into a local-set-based one, which can be locally determined, although the global condition may provide a sharper estimation.

Besides, in each iteration of ILSR and IWR, all the vertices only use information associated with themselves, while data has to be transmitted from the sampled vertices to their neighbors in IPR method. As a result, the former two methods may be easier applied in potential distributed scenario \cite{wang_distributed_2014}.

\section{Discussions on Local Sets}\label{discusslocalset}

To accelerate the convergence of signal reconstruction algorithms, a weighting (in IWR) or propagating (in IPR) procedure is introduced in the existing algorithm ILSR. Both the procedures are based on a division of the graph, i.e., local sets.
In this section, we will firstly show two special sampling set and local sets, then discuss the choice of local sets for general cases.

\subsection{Special Sampling Set and Local Sets}

If the sampling set $\mathcal{S}=\mathcal{V}$, then $|\mathcal{N}(u)|=1$, $K(u)=0$, and $Q_{\text{max}}=0$. According to Proposition \ref{pro2} and Proposition \ref{cor1}, any ${\bf f}$ satisfying $\omega<\infty$ can be reconstructed by IWR and IPR, which is a natural result.

Another extreme case is the sampling set $\mathcal{S}$ contains only one vertex and the corresponding local set is all the vertices in the graph.
In this case, only constant signals can be reconstructed from the sampled data, which means that only one discrete frequency $\omega=0$ can satisfy the condition $\omega<1/Q_{\rm max}^2$.
It is easy to understand because only the signals with the same value for all the vertices can be reconstructed from only a single sample.
The analysis above is always true no matter which vertex is chosen as the sampled one.
Therefore, the following corollary gives an estimation of the smallest positive eigenvalue of a graph.
\begin{cor}
For a graph Laplacian, its smallest positive eigenvalue $\lambda_{\text{min}}$ satisfies
$$
\lambda_{\text{min}}\ge \max_{u\in\mathcal{V(G)}}\frac{1}{K(u)R(u)},
$$
where $K(u)$ and $R(u)$ are defined in Definitions \ref{defi1} and \ref{radius}, respectively, in which the local set contains all the vertices of graph $\mathcal{G}$, i.e., $\mathcal{G}_{\mathcal{N}(u)}=\mathcal{G}$.
\end{cor}

\subsection{On the Evaluation of Local Sets}\label{discussion}
According to the sufficient condition in Proposition \ref{pro2} and \ref{cor1}, a sampling set and the associated local sets with a smaller $Q_{\text{max}}$ usually lead to a wider range of bandlimited signal which can be guaranteed to reconstruct. Besides, for a given $\omega$, a smaller $Q_{\text{max}}$ leads to a better error bound of convergence, i.e., a smaller $\gamma$ (for IPR) or $2\gamma/(1+\gamma^2)$ (for IWR). Therefore, when the graph topology is given, it is necessary to find a proper sampling set $\mathcal{S}$ and the corresponding vertex division $\{\mathcal{N}(u)\}_{u\in\mathcal{S}}$, which makes the quantity $Q_{\text{max}}$ as small as possible.

However, since the sufficient condition we give is rather conservative and not very sharp, minimizing $Q_{\text{max}}$ is only a rough way to obtain a better division of local sets.
In other words, there may be some better evaluation of local sets than $Q_{\text{max}}$.
Finding the optimal division of local sets is still an open problem and needs more comprehensive study.
Therefore we have not focused on how to construct local sets to minimize $Q_{\text{max}}$ in this paper.
It may be done better when a sharper sufficient condition is provided, which will be studied in the future work.
In this paper, only a special choice of sampling set and the associated local sets with $Q_{\text{max}}=1$ are presented in the following text.

\subsection{A Special Case of One-hop Sampling}
\label{onehopsection}

In this section, we present a special case of \emph{dense} sampling where all entries to be recovered are directly connected to the sampled vertices. One may read that such dense sampling facilitates the local sets partition.

\begin{cor}\label{cor3}
For a given sampling set $\mathcal S$, if the local sets $\{\mathcal{N}(u)\}_{u\in\mathcal{S}}$ satisfies
\begin{equation}\label{onehop}
\max_{v\in \mathcal{N}(u)}\text{dist}(u,v)\le 1,\quad \forall u\in\mathcal{S},
\end{equation}
the sufficient condition of recovery in Proposition \ref{pro2} and Proposition \ref{cor1} can be refined as
$\omega<1$ and $\gamma=\sqrt{\omega}$.
\end{cor}

\begin{proof}
The condition (\ref{onehop}) means that all the vertices except $u$ in $\mathcal{N}(u)$ are connected to $u$. It implies that $d_{\mathcal{N}(u)}(u)=|\mathcal{N}(u)|-1$ and then $\tilde{K}(u)\le1$. Besides, it is obvious to see $R(u)\le1$. Therefore, $\tilde{Q}_{\text{max}}=1$ and Corollary \ref{cor3} is obtained.
\end{proof}

A greedy method is proposed and described in Table \ref{algOneHop}, which can produce the one-hop sampling set and the associated local sets at the same time and satisfy the condition of (\ref{onehop}).
The reason for selecting the vertex with the largest degree and its neighbors is that more vertices can be removed in each step, which may lead to a sampling set with fewer vertices.
One may accept that this is a rather economical choice of sampling set when there is no restriction on the number or location of the sampling vertices, because both $K(u)$ and $R(u)$ are small simultaneously.

\begin{table}[t]
\renewcommand{\arraystretch}{1.2}
\caption{A Greedy Method for a One-hop Sampling Set.}\label{algOneHop}
\begin{center}
\begin{tabular}{l}
\toprule[1pt]
{\bf Input:} \hspace{0.5em} Graph $\mathcal{G(V,E)}$;\\
{\bf Output:} \hspace{0.5em} One-hop sampling set $\mathcal{S}$, local sets $\{\mathcal{N}(u)\}_{u\in\mathcal{S}}$;\\
\hline
{\bf Initialization:}\hspace{0.5em}$\mathcal{S}=\emptyset$;\\
{\bf Loop:}\\
\hspace{1.3em} 1) Find the largest-degree vertex, $\displaystyle u=\arg\max_{v\in \mathcal{V}}d_{\mathcal{G}}(v)$;\\
\hspace{1.3em} 2) Add $u$ into the sampling set, $\displaystyle \mathcal{S}=\mathcal{S}\cup \{u\}$;\\
\hspace{1.3em} 3) The one-hop local set $\displaystyle \mathcal{N}(u)=\{u\}\cup\{v\in\mathcal{V}|(u,v)\in\mathcal{E}\}$;\\
\hspace{1.3em} 4) Remove the edges, $\displaystyle \mathcal{E}=\mathcal{E}\backslash\{(p,q)|p\in\mathcal{N}(u),q\in\mathcal{V}\}$;\\
\hspace{1.3em} 5) Remove the vertices, $\displaystyle \mathcal{V}=\mathcal{V}\backslash\mathcal{N}(u)$ and $\mathcal{G=G(V,E)}$;\\
{\bf Until:}\hspace{0.5em} $\mathcal{V}=\emptyset$.\\
\bottomrule[1pt]
\end{tabular}
\end{center}
\end{table}

\section{Relationship with Time Domain Results}\label{sec6}

\begin{figure*}
\begin{center}
\includegraphics[width=\textwidth]{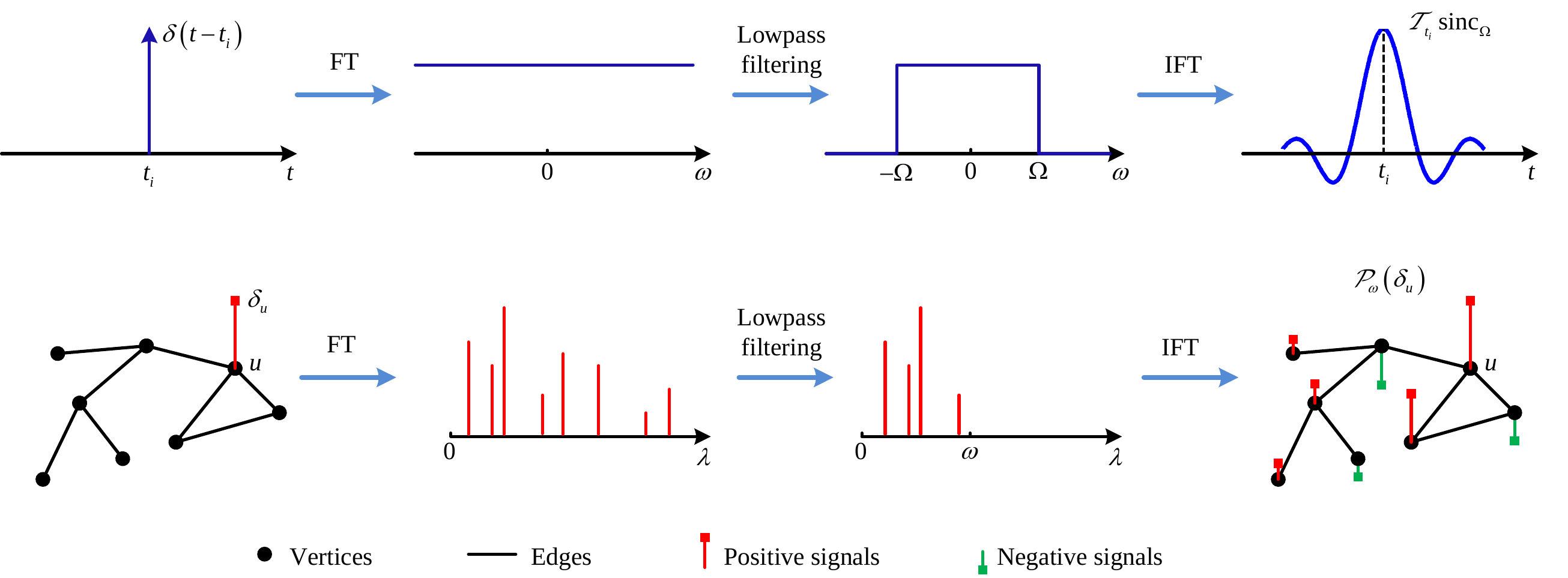}
\caption{The correspondence between irregular sampling in the time domain and that on graph.
Both $\mathcal{T}_{t_i}{\rm sinc}_{\Omega}$ and $\mathcal{P}_{\omega}(\bm{\delta}_u)$ are the projections of $\delta$-functions onto bandlimited spaces.
Under certain conditions, for all the sampled points or vertices, these kinds of signals $\{\mathcal{T}_{t_i}{\rm sinc}_{\Omega}\}_{t_i\in\mathcal{S}}$
and $\{\mathcal{P}_{\omega}(\bm{\delta}_u)\}_{u\in\mathcal{S}}$ become frames. Consequently, the original signals can be reconstructed by the sampled data.}
\label{TimeAndVertex}
\end{center}
\end{figure*}

Bandlimited signal sampling and reconstruction on graph is closely related to irregular sampling \cite{feichtinger_theory_1994,grochenig_adiscrete_1993}
or nonuniform sampling \cite{marvasti_nonuniform_2001} in the time domain,
which sheds light on the analysis of graph signal.
There have existed several iterative reconstruction methods and theoretically analysis of time-domain irregular sampling
\cite{sauer_iterative_1987,marvasti_recovery_1991,grochenig_reconstruction_1992,feichtinger_theory_1994},
some of which are related to the frame theory \cite{marvasti_recovery_1991,benedetto_irregular_1992,feichtinger_theory_1994}.
Some further works extend the results to high dimensional spaces \cite{feichtinger_theory_1994} and manifolds \cite{pesenson_poincare_2004,feichtinger_recovery_2004}.

By exploiting the similarities between time-domain irregular sampling and graph signal sampling, some results of this work have consistent formulation with the corresponding results in the time domain. The reconstruction methods also have correspondences in the time domain.

Results on time-domain irregular sampling show that $\{\mathcal{T}_{t_i}{\rm sinc}_{\Omega}\}_{t_i\in\mathcal{S}}$ is a frame for $\mathcal{B}_{\Omega}^2$
if the sampling set $\mathcal{S}$ satisfies some particular conditions,
where $\mathcal{T}_{t_i}f(t)=f(t-t_i)$ denotes the translation of $f(t)$, ${\rm sinc}_{\Omega}$ denotes the sinc function whose bandwidth is $\Omega$
and $\mathcal{B}_{\Omega}^2$ denotes the space of $\Omega$-bandlimited square integrable signal.

Correspondingly, for the graph signal ${\mathbf f}\in PW_{\omega}(\mathcal{G})$, $\{\mathcal{P}_{\omega}(\bm{\delta}_u)\}_{u\in\mathcal{S}}$ is a frame under some conditions.
The result is consistent with that in the time domain.
The correspondence between irregular sampling in the time domain and that on graph is illustrated in Fig. \ref{TimeAndVertex}.
In the graph signal sampling problem, $\{\mathcal{P}_{\omega}(\bm{\delta}_u)\}_{u\in\mathcal{S}}$ corresponds to the frame $\{\mathcal{T}_{t_i}{\rm sinc}_{\Omega}\}_{t_i\in\mathcal{S}}$ in the time domain.
The essence of the two problem is very similar and theoretical results on sampling and reconstruction of graph signals can be obtained enlightened by irregular sampling in the time domain.

The ideas of graph signal reconstruction methods ILSR, IWR and IPR have correspondences in time-domain, which are Marvasti method \cite{marvasti_recovery_1991}, adaptive weights method \cite{feichtinger_theory_1994} and Voronoi method \cite{grochenig_reconstruction_1992}, respectively.
However, a graph is discrete and the local topology of each sampling vertex is irregular, which leads to some new problems related to local sets in the sampling and reconstruction of graph signals.

The correspondence between time-domain irregular sampling and graph signal sampling is shown in Table \ref{tab1}.

\begin{table}[h!t]
\caption{The correspondence between irregular sampling in the time domain and that on graph.}\label{tab1}
    \begin{center}
        \begin{tabular}{ccc}
\toprule[1pt]
            Terms & Time Domain & Vertex Domain \\ \hline
            Signal & $f(t)$ & ${\bf f}$ \\
            Cutoff frequency & $\Omega$ & $\omega$ \\
            Low-frequency space & $\mathcal{B}_{\Omega}^2$ & $PW_{\omega}(\mathcal{G})$ \\
            Shifted impulse & $\delta(t-t_i)$ & $\bm{\delta}_u$ \\
            Shifted sinc function & $\mathcal{T}_{t_i}{\rm sinc}_{\Omega}$ & $\mathcal{P}_{\omega}(\bm{\delta}_u)$ \\
            Neighborhood & $[(t_{i-1}+t_i)/2,(t_i+t_{i+1})/2)$ & $\mathcal{N}(u)$ \\
            Neighbor indicator & $1_{[(t_{i-1}+t_i)/2,(t_i+t_{i+1})/2)}$ & $\bm{\delta}_{\mathcal{N}(u)}$ \\
            Weight & $\sqrt{(t_{i+1}-t_{i-1})/2}$ & $\sqrt{|\mathcal{N}(u)|}$ \\
            $\!\!\!\!$Reconstruction method & Marvasti Method & ILSR \\
            $\!\!\!\!$Reconstruction method & Adaptive Weights Method & IWR \\
            $\!\!\!\!$Reconstruction method & Voronoi Method & IPR \\
\bottomrule[1pt]
        \end{tabular}
    \end{center}
\end{table}

\section{Experiments}\label{sec7}
The Minnesota road graph \cite{gleich_matlabbgl} is chosen as the graph, which has $2640$ vertices and $6604$ edges, to test the proposed reconstruction algorithms. The bandlimited signal is generated by first generating a random signal and then removing its high-frequency components.

\subsection{Convergence Rate}
The convergence rate of the three algorithms are compared here.
A one-hop sampling set satisfying (\ref{onehop}) is chosen as the sampling set.
The one-hop sampling set and the corresponding local sets are obtained by the greedy method in Table \ref{algOneHop}.
This sampling set has $872$ vertices, which is about one third of all.
The cutoff frequency is $0.25$.
The convergence curves of ILSR, IWR, and IPR are illustrated in Fig. \ref{exp2}.
It is obvious that the convergence rate of the proposed algorithms is significantly improved compared with the reference.
Furthermore, IPR is better than IWR on the convergence rate.
Both observations are in accordance with the analysis in \ref{intuitive}.

\begin{figure}[t]
\begin{center}
\includegraphics[width=9cm]{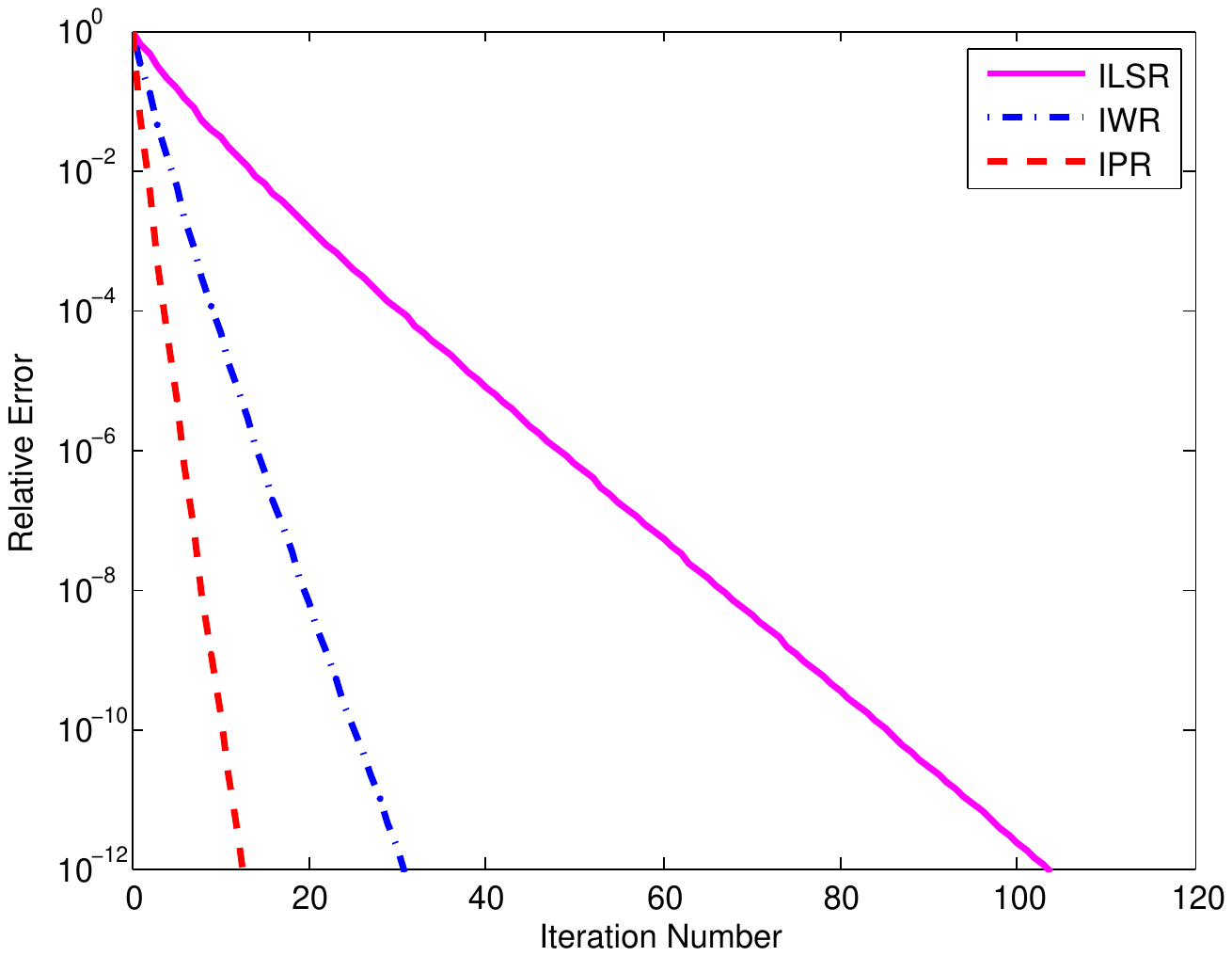}
\caption{Convergence curves of ILSR, IWR, and IPR.}
\label{exp2}
\end{center}
\end{figure}

\subsection{Sampling Geometry}\label{EXPSG}

The choice of sampling set may affect the performance of convergence.
Two different sampling sets are used to reconstruct the same bandlimited original signal. Both of the sampling sets have the same amount of vertices.
The first sampling set is the one-hop set satisfying (\ref{onehop}), with $872$ sampled vertices and $Q_{\rm{max}}=1$.
For the second sampling set, $872$ vertices are selected uniformly at random among all the vertices.
Each unsampled vertex belongs to the local set associated with its nearest sampled vertex.
Then $K(u)$ and $R(u)$ can be obtained for each local set, and we have $Q_{\rm{max}}=\sqrt{40}$ with the corresponding $K(u)=8$ and $R(u)=5$.
The convergence curves of the two sampling sets using the three reconstruction methods are illustrated in Fig. \ref{exp3}.
For all the algorithms, the convergence is faster by using the sampled data of the one-hop sampling set than by using the randomly chosen vertex set.
It means that the sampling geometry has influence on the reconstruction.
A sampling set and the local sets with smaller $Q_{\rm{max}}$ may converge faster.

\begin{figure}[ht]
\begin{center}
\includegraphics[width=9cm]{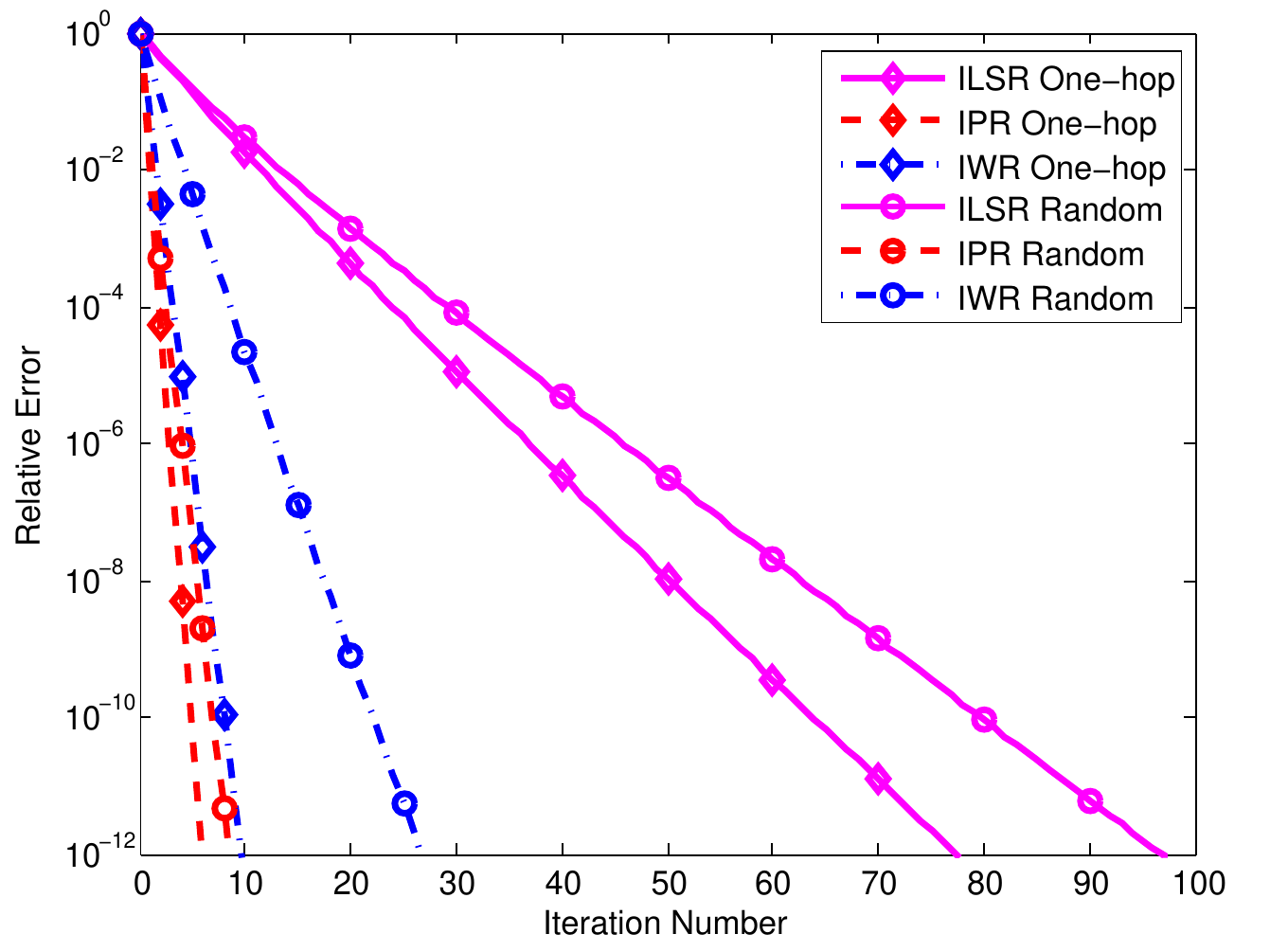}
\caption{Convergence curves of ILSR, IWR, and IPR, on one-hop sampling set and randomly chosen sampling set.}
\label{exp3}
\end{center}
\end{figure}

\subsection{Actual and Priori Known Cutoff Frequencies}

\begin{figure}[ht]
\begin{center}
\includegraphics[width=9cm]{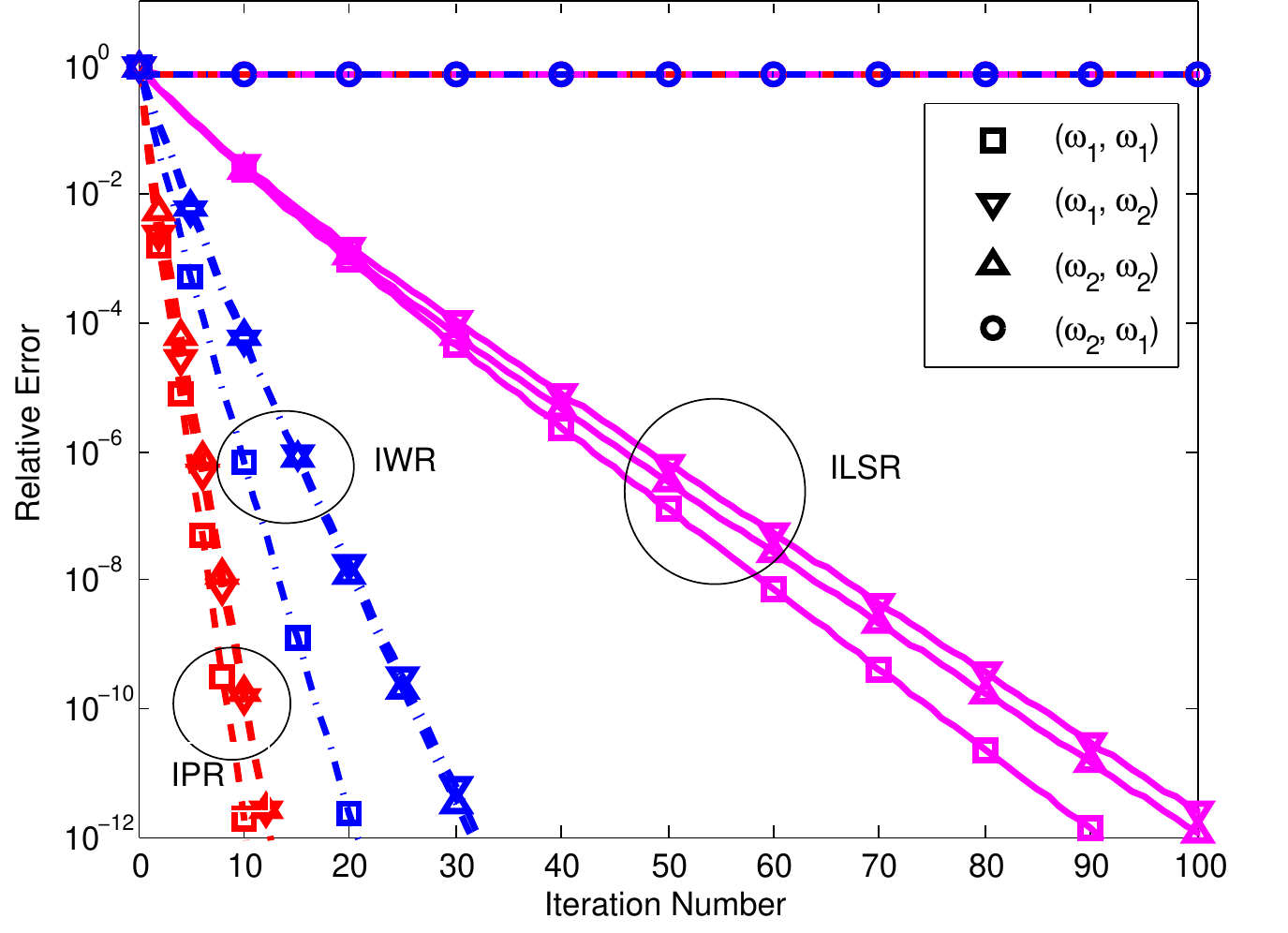}
\caption{Convergence curves of three cases with different actual and priori known cutoff frequencies.}
\label{exp5}
\end{center}
\end{figure}

The cutoff frequency is a crucial quantity in the reconstruction of the bandlimited signal.
For a bandlimited signal the cutoff frequency is known as a priori knowledge.
However, the priori knowledge may be an estimate rather than ground truth. In this experiment, the effect on the imprecise knowledge of the cutoff frequency is investigated.
For frequencies $\omega_1$ and $\omega_2$ satisfying $\omega_1<\omega_2$, the following four cases are considered:
1)$(\omega_1, \omega_1)$; 2)$(\omega_1, \omega_2)$; 3)$(\omega_2, \omega_2)$; 4)$(\omega_2, \omega_1)$, where $(\omega_i, \omega_j)$ means the actual cutoff frequency is $\omega_i$ and the priori known frequency is $\omega_j$.

In the experiment we set $\omega_2=2\omega_1$.
The convergence curves are illustrated in Fig. \ref{exp5}.
It is easy to understand that the relative error of case 4) is large because only about half of the energy can be preserved.
For case 1) and 2), although both the original signals are actually $\omega_1$-bandlimited, the reconstruction converges faster in case 1) than case 2) because of a more accurate priori knowledge.
Comparing case 2) and 3), it can be seen that the convergence curves almost coincide, which means that although the bandwidth of the original signal is reduced, the convergence rate may increase little if the reduction is not priori known.
The experimental results show that the convergence rate depends little on the actual cutoff frequency but depends more on the cutoff frequency the signal is regarded to have.
If we have more accurate priori knowledge on the cutoff frequency, the reconstruction will be more efficient.

In fact, the $\omega$ of $PW_{\omega}(\mathcal{G})$ is the priori known cutoff frequency, rather than the actual one.
Even though the actual cutoff frequency of the original signal is smaller than $\omega$, the decay coefficient $\gamma =Q_{\rm max}\sqrt{\omega}$ is determined by $\omega$, i.e., the cutoff frequency of the subspace.
In other words, convergence is a property of the frame, which is determined by the low-frequency subspace, and not of the signal we are trying to reconstruct.

\subsection{Theoretical and Numerical Bounds for Cutoff Frequency}
The given sufficient condition for the convergence of IWR and IPR is rather conservative and not very sharp for all the graphs.
This experiment shows the actual cutoff frequency that the reconstruction algorithms can recover.
In this experiment, the sampling set and local sets are the same as that in \ref{EXPSG} with $Q_{\rm{max}}=\sqrt{40}$.
The experimental result is illustrated in Fig. \ref{exp13}.
For each cutoff frequency, $100$ signals within the subspace are generated randomly. The curves show the rate of signals that converge within a relative error $10^{-3}$ in $20$ iterations.
For this sampling set and local sets, the sufficient condition we provide is $\omega<0.025$. It can be seen that the reconstruction methods work in a larger low-frequency subspace, which means there is still room for improvement to give a better bound.

\begin{figure}[ht]
\begin{center}
\includegraphics[width=9cm]{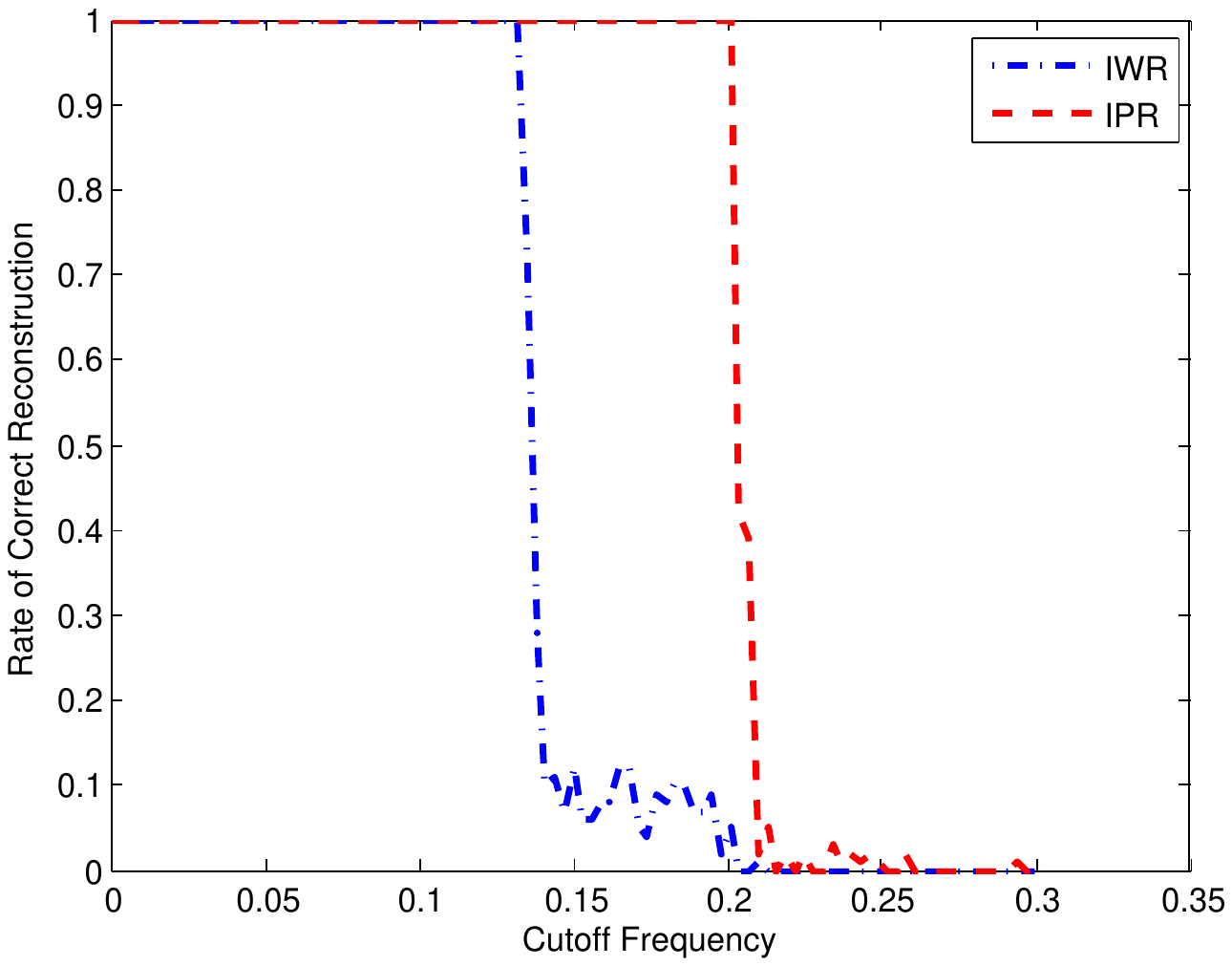}
\caption{The rate of signals that converge within a relative error $10^{-3}$ in $20$ iterations.}
\label{exp13}
\end{center}
\end{figure}

\subsection{Robustness against Noise}
\subsubsection{Observation Noise}
Suppose there is noise involved in the observation of sampled graph signal. This experiment focuses on the robustness to the observation noise of the three algorithms.
In this experiment the noise is generated as independent identical distributed Gaussian sequence.
As shown in Fig. \ref{exp6}, the steady-state error decreases as the SNR increases.  The three methods have almost the same performance against observation noise.

\begin{figure}[ht]
\begin{center}
\includegraphics[width=9cm]{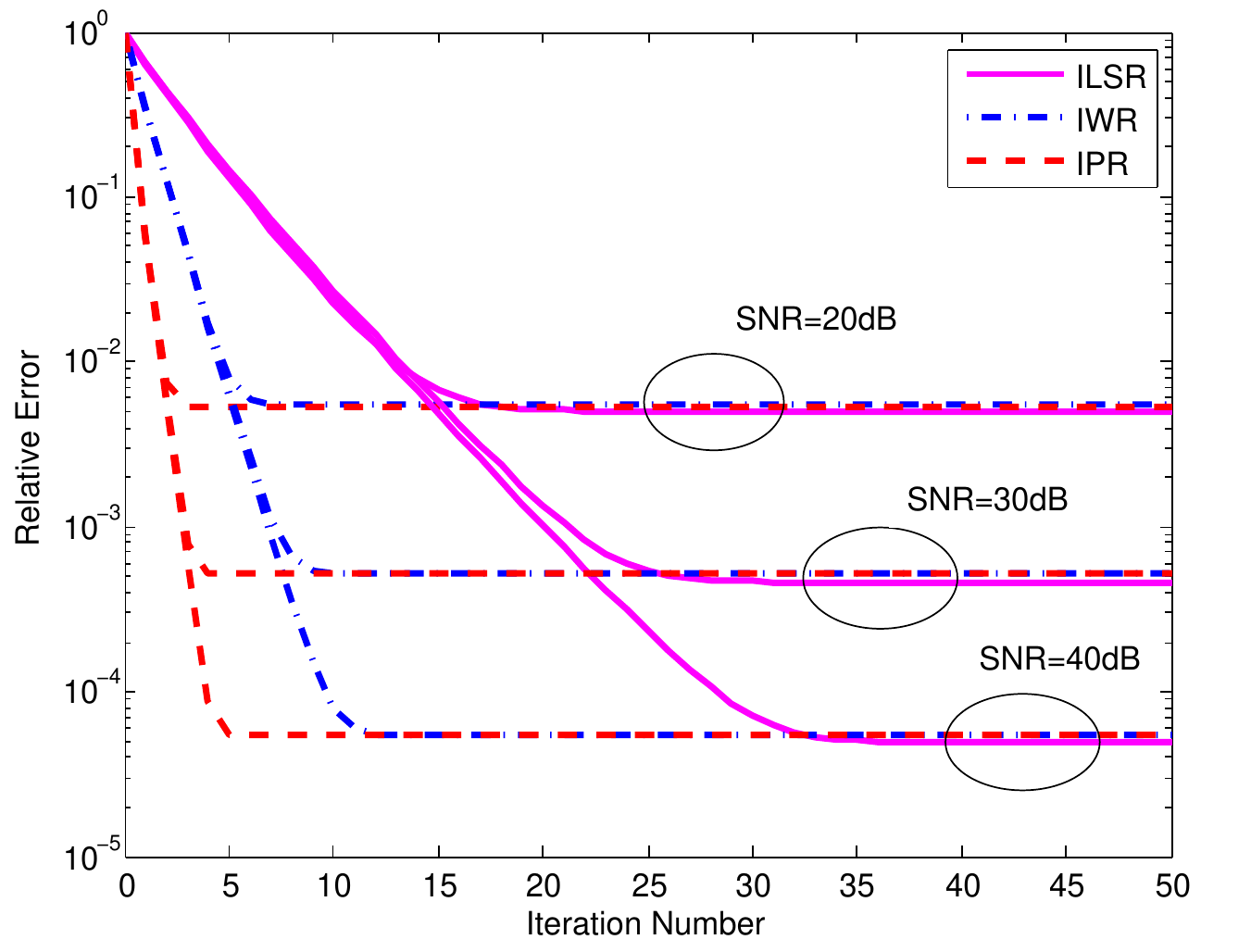}
\caption{Convergence curves of three reconstruction methods with various observation SNR.}
\label{exp6}
\end{center}
\end{figure}

\subsubsection{Reconstruction of Approximated Bandlimited Signals}
\begin{figure}[ht]
\begin{center}
\includegraphics[width=9cm]{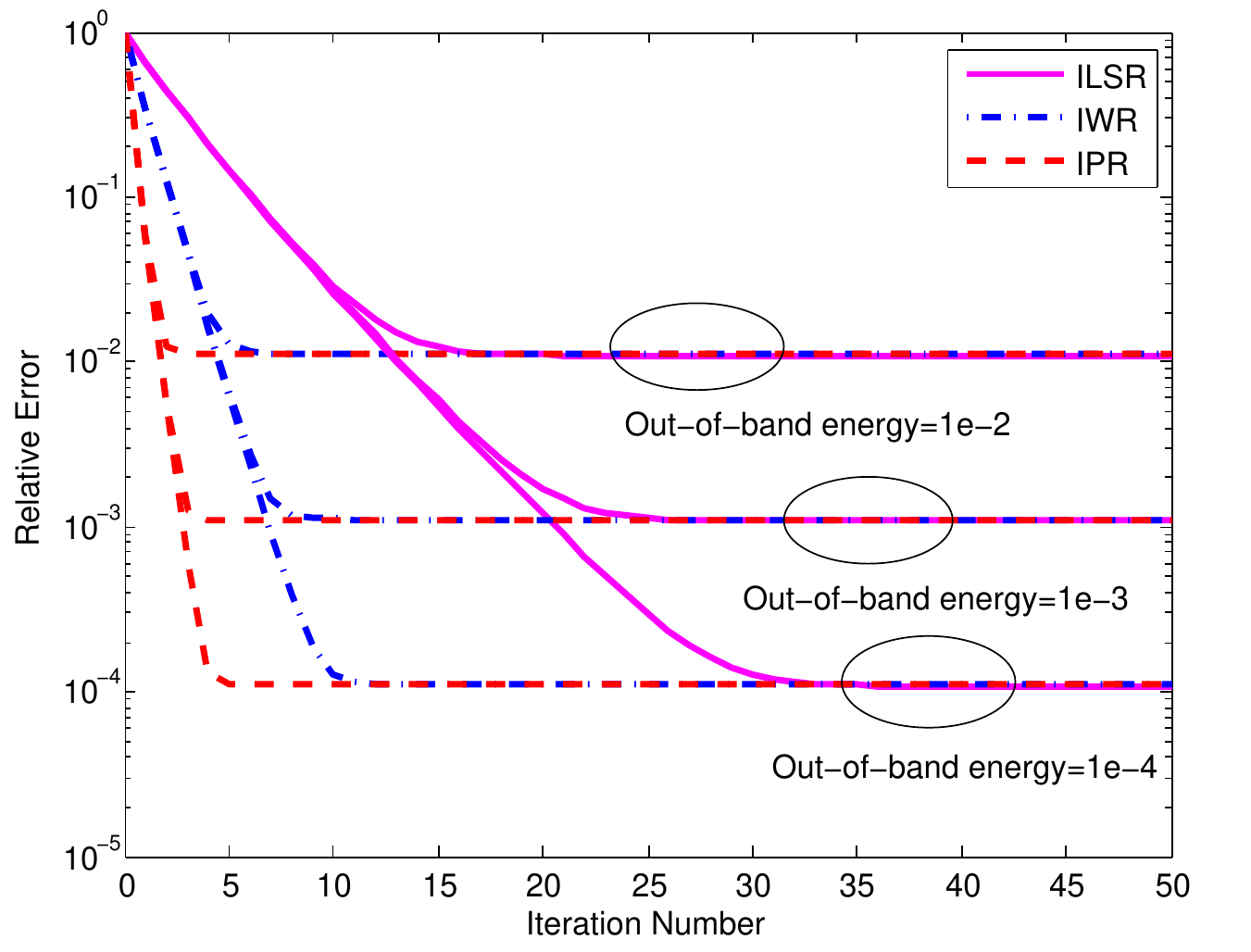}
\caption{Convergence curves of three reconstruction methods for approximated bandlimited signal.}
\label{exp8}
\end{center}
\end{figure}

Real-world data is always not strictly bandlimited. However, most smooth signals over graph can be regarded as approximated bandlimited signals.
In the experiment in Fig. \ref{exp8}, the three methods are used to reconstruct signals with different out-of-band energy.
The steady-state error will be larger for signals with more energy out of band. Besides, the three algorithms perform almost the same for approximated bandlimited signals.

\section{Conclusion}\label{sec8}

In this paper, the problem of graph signal reconstruction in bandlimited space is studied.
We first propose a concept of local set, where all vertices in the graph are divided into disjoint local sets associated with the sampled vertices. Based on frame theory, an operator named local propagation is then proposed and proved to be contraction mapping. Consequently, several series of signals are proved to be frames and their frame bounds are estimated.
Above theory provides solid foundation for developing efficient reconstruction algorithms. Two local-set-based iterative methods called IWR and IPR are proposed to reconstruct the missing data from the observed samples. Strict proofs of convergence and error bounds of IWR and IPR are presented.
After comprehensive discussion on the proposed algorithms, we explore the correspondence between time-domain irregular sampling and graph signal sampling, which sheds light on the analysis in the graph vertex domain.
Experiments, which verify the theoretical analysis, show that IPR performs beyond IWR, and both the proposed methods converge significantly faster than the reference algorithm.

\section{Appendix}

\subsection{Proof of Lemma \ref{lemma1}}\label{proof1}
\begin{proof}
By the definition of ${\bf G}$, one has
\begin{align}\label{lem1-1}
\|{\mathbf f}-{\bf G}{\mathbf f}\|^2=&\left\|\mathcal{P}_{\omega}\left({\mathbf f}-\sum_{u\in \mathcal{S}}f(u)\bm{\delta}_{\mathcal{N}(u)}\right)\right\|^2\nonumber\\
\le&\left\|{\mathbf f}-\sum_{u\in \mathcal{S}}f(u)\bm{\delta}_{\mathcal{N}(u)}\right\|^2\nonumber\\
\le&\sum_{u\in \mathcal{S}}\sum_{v\in \mathcal{N}(u)}|f(v)-f(u)|^2.
\end{align}
Considering that $\mathcal{N}(u)$ is connected, there is always a shortest path within $\mathcal{N}(u)$ from any $v\in\mathcal{N}(u)$ to $u$, which is denoted as $(v, v_1, \cdots, v_{k_v}, u)$. One has
\begin{align}\label{lem1-3}
|f(v)-f(u)|^2
\le R(u)\left(|f(v)-f(v_1)|^2+\cdots +|f(v_{k_v})-f(u)|^2\right),
\end{align}
which is because any path is not longer than $R(u)$.

For each $v$ satisfying $(u,v)\in \mathcal{E}(\mathcal{T}(u))$, the path from any vertex in $\mathcal{T}_u(v)$ to $u$ contains edge $(u,v)$ and this edge is counted for $|\mathcal{T}_u(v)|$ times. By the definition of $K(u)$, each edge in $\mathcal{N}(u)$ is counted for no more than $K(u)$ times.
Then,
\begin{align}\label{lem1-2}
\sum_{v\in \mathcal{N}(u)}\!\!\!|f(v)-f(u)|^2\le K(u)R(u)\!\!\!\!\!\sum_{\substack{(p, q)\in\mathcal{E}\\ p, q\in \mathcal{N}(u)}}\!\!\!\!\!|f(p)-f(q)|^2,
\end{align}

By the assumption of $\omega$-bandlimited signal, the following inequality is established
\footnote{Lemma 2.1 of \cite{fuhr_poincare_2013} has proved a more general case for weighted graphs.}.
\begin{align}\label{lem1-4}
\sum_{(p,q)\in\mathcal{E}}|f(p)-f(q)|^2
=&\sum_{p\in \mathcal{V}}d(p)|f(p)|^2-2\sum_{(p,q)\in\mathcal{E}}f(p)f(q)\nonumber\\
=&{\mathbf f}^{\rm T}{\bf L}{\mathbf f}={\mathbf f}^{\rm T}{\bf V\Lambda V}^{\rm T}{\mathbf f}=\hat{{\mathbf f}}^{\rm T}{\bf \Lambda} \hat{{\mathbf f}}\nonumber\\
=&\sum_{\lambda_i\le\omega}\lambda_i |\hat{f}(i)|^2\le \omega\hat{{\mathbf f}}^{\rm T}\hat{{\mathbf f}}=\omega\|{\mathbf f}\|^2.
\end{align}
In the above derivation, $d(p)$ denotes the degree of vertex $p$, and $\hat{{\mathbf f}}$ denotes the graph Fourier transform of ${\bf f}$. The last inequality is because the components of $\hat{{\mathbf f}}$ corresponding to the frequencies higher than $\omega$ are zero for ${\mathbf f}\in PW_{\omega}(\mathcal{G})$.

Combining (\ref{lem1-1}), (\ref{lem1-2}), and (\ref{lem1-4}), we have
\begin{align}
\|{\mathbf f}-{\bf G}{\mathbf f}\|^2\le&\sum_{u\in \mathcal{S}}\left(K(u)R(u)\!\!\!\!\sum_{\substack{(p, q)\in\mathcal{E}\\ p,q\in \mathcal{N}(u)}}\!\!\!\!|f(p)-f(q)|^2\right)\nonumber\\
\le&Q_{\text{max}}^2\sum_{(p, q)\in\mathcal{E}}|f(p)-f(q)|^2\nonumber\\
\le& Q_{\text{max}}^2\omega\|{\mathbf f}\|^2\nonumber
\end{align}
and Lemma \ref{lemma1} is proved.
\end{proof}

\subsection{Proof of Proposition \ref{pro3}}\label{proof3}
\begin{proof}
By the definition of local propagation, $\forall {\mathbf f}\in PW_{\omega}(\mathcal{G})$, one has
\begin{align}
{\bf G}{\mathbf f}=&\sum_{u\in \mathcal{S}}\langle {\mathbf f}, \bm{\delta}_u\rangle \mathcal{P}_{\omega}(\bm{\delta}_{\mathcal{N}(u)})\nonumber\\
=&\sum_{u\in \mathcal{S}}\langle \mathcal{P}_{\omega}({\mathbf f}), \bm{\delta}_u\rangle \mathcal{P}_{\omega}(\bm{\delta}_{\mathcal{N}(u)})\nonumber\\
=&\sum_{u\in \mathcal{S}}\langle {\mathbf f}, \mathcal{P}_{\omega}(\bm{\delta}_u)\rangle \mathcal{P}_{\omega}(\bm{\delta}_{\mathcal{N}(u)}).\label{eq1inproof3}
\end{align}
Utilizing \eqref{eq1inproof3} in Lemma \ref{lemma1}, one gets
\begin{equation}\label{lem2-1}
\left\|{\mathbf f}-\sum_{u\in \mathcal{S}}\langle {\mathbf f}, \mathcal{P}_{\omega}(\bm{\delta}_u)\rangle \mathcal{P}_{\omega}(\bm{\delta}_{\mathcal{N}(u)})\right\|\le \gamma\|{\mathbf f}\|.
\end{equation}
For all ${\mathbf f}\in PW_{\omega}(\mathcal{G})$ and $\{g_u\}_{u\in \mathcal{S}}$, we have
\begin{align}\label{lem2-2}
\sum_{u\in \mathcal{S}}|\langle {\mathbf f}, \mathcal{P}_{\omega}(\bm{\delta}_u)\rangle|^2
=& \sum_{u\in \mathcal{S}}|f(u)|^2\le\|{\mathbf f}\|^2.
\end{align}
and
\begin{align}\label{lem2-3}
\left\|\sum_{u\in \mathcal{S}}g_u \mathcal{P}_{\omega}(\bm{\delta}_{\mathcal{N}(u)})\right\|^2
=&\left\|\mathcal{P}_{\omega}\left(\sum_{u\in \mathcal{S}}g_u\bm{\delta}_{\mathcal{N}(u)}\right)\right\|^2\nonumber\\
\le& \left\|\sum_{u\in \mathcal{S}}g_u\bm{\delta}_{\mathcal{N}(u)}\right\|^2\nonumber\\
=&\sum_{u\in \mathcal{S}}|\mathcal{N}(u)|\cdot|g_u|^2\nonumber\\
\le& N_{\text{max}}\sum_{u\in \mathcal{S}}|g_u|^2.
\end{align}
Combining (\ref{lem2-1}), (\ref{lem2-2}) and (\ref{lem2-3}) and Proposition 2 in \cite{feichtinger_theory_1994}
\footnote{
Proposition 2 in \cite{feichtinger_theory_1994}: Suppose $\{{\bf e}_n\}$ and $\{{\bf h}_n\}$ satisfy that there exist constant $C_1, C_2>0$ and $0\le \gamma<1$, so that
$\sum|\langle {\bf f}, {\bf e}_n\rangle|^2\le C_1\|{\bf f}\|^2$, $\|\sum \lambda_n{\bf h}_n\|^2\le C_2\sum |\lambda_n|^2$ and
$\|{\bf f}-\sum\langle {\bf f}, {\bf e}_n\rangle{\bf h}_n\|\le\gamma\|{\bf f}\|$ for all ${\bf f}\in \mathcal{H}$ and $\{\lambda_n\}$.
Then $\{{\bf e}_n\}$ is a frame with frame bounds $(1-\gamma)^2/C_2$ and $C_1$, and $\{{\bf h}_n\}$ is a frame with bounds $(1-\gamma)^2/C_1$ and $C_2$.
},
$\{\mathcal{P}_{\omega}(\bm{\delta}_{\mathcal{N}(u)})\}_{u\in \mathcal{S}}$ is a frame with bounds $(1-\gamma)^2$ and $N_{\text{max}}$,
and $\{\mathcal{P}_{\omega}(\bm{\delta}_u)\}_{u\in \mathcal{S}}$ is a frame with bounds $(1-\gamma)^2/N_{\text{max}}$ and $1$.
Proposition \ref{pro3} is proved.
\end{proof}

\subsection{Proof of Lemma \ref{lemma3}}
\label{proof4}

\begin{proof}
According to Lemma \ref{lemma1} and Proposition \ref{pro4}, we have $\|{\bf I-G}\|\le \gamma<1$ for $PW_{\omega}(\mathcal{G})$ when $\gamma=Q_{\text{max}}\sqrt{\omega}<1$.
Then ${\bf G}$ is invertible and $1-\gamma\le \|{\bf G}\|\le 1+\gamma$ for $PW_{\omega}(\mathcal{G})$.
\begin{align}
\left\|{\mathbf f}\right\|^2=&\|{\bf G}^{-1}{\bf G}{\mathbf f}\|^2\nonumber\\
\le& (1-\gamma)^{-2}\|{\bf G}{\mathbf f}\|^2\nonumber\\
\le& (1-\gamma)^{-2}\left\|\sum_{u\in \mathcal{S}}f(u)\bm{\delta}_{\mathcal{N}(u)}\right\|^2\nonumber\\
=& (1-\gamma)^{-2}\sum_{u\in \mathcal{S}}|\mathcal{N}(u)|\cdot|f(u)|^2.\nonumber
\end{align}
Then the left inequality of Lemma \ref{lemma3} is proved.

From the proof of Lemma \ref{lemma1}, it is known that
$$
\left\|{\mathbf f}-\sum_{u\in \mathcal{S}}f(u)\bm{\delta}_{\mathcal{N}(u)}\right\|\le \gamma\|{\mathbf f}\|.
$$
Therefore,
\begin{align}
\sum_{u\in \mathcal{S}}|\mathcal{N}(u)|\cdot|f(u)|^2
=&\left\|\sum_{u\in \mathcal{S}}f(u)\bm{\delta}_{\mathcal{N}(u)}\right\|^2\nonumber\\
\le&\left(\|{\mathbf f}\|+\left\|{\mathbf f}-\sum_{u\in \mathcal{S}}f(u)\bm{\delta}_{\mathcal{N}(u)}\right\|\right)^2 \nonumber\\
\le& (1+\gamma)^2\|{\mathbf f}\|^2,\nonumber
\end{align}
which is the right inequality of Lemma \ref{lemma3}.

Considering
\begin{align}
|\langle {\mathbf f}, \sqrt{|\mathcal{N}(u)|}\mathcal{P}_{\omega}(\bm{\delta}_{u})\rangle|^2
&=|\mathcal{N}(u)|\cdot|\langle \mathcal{P}_{\omega}({\mathbf f}),\bm{\delta}_{u}\rangle|^2\nonumber\\
&=|\mathcal{N}(u)|\cdot|f(u)|^2,\nonumber
\end{align}
The inequalities imply that $\{\sqrt{|\mathcal{N}(u)|}\mathcal{P}_{\omega}(\bm{\delta}_{u})\}_{u\in \mathcal{S}}$ is a frame for $PW_{\omega}(\mathcal{G})$ with bounds
$\left(1-\gamma\right)^2$ and $\left(1+\gamma\right)^2$.
\end{proof}

\subsection{Proof of Proposition \ref{pro2}}\label{proof5}
\begin{proof}
From Lemma \ref{lemma3}, $\{\sqrt{|\mathcal{N}(u)|}\mathcal{P}_{\omega}(\bm{\delta}_{u})\}_{u\in \mathcal{S}}$ is a frame with bounds
$A=(1-\gamma)^2$ and $B=(1+\gamma)^2$.
By the property of frame \cite{Christensen_an_2002}, the original signal can be reconstructed by
$$
{\mathbf f}^{(k+1)}={\mathbf f}^{(k)}+\mu {\bf G}_{\text{w}}({\mathbf f}-{\mathbf f}^{(k)})
$$
where the frame operator is
\begin{align}
{\bf G}_{\text{w}} {\mathbf f}=&\sum_{u\in \mathcal{S}}\langle {\mathbf f}, \sqrt{|\mathcal{N}(u)|}\mathcal{P}_{\omega}(\bm{\delta}_{u})\rangle\sqrt{|\mathcal{N}(u)|}\mathcal{P}_{\omega}(\bm{\delta}_{u})\nonumber\\
=&\mathcal{P}_{\omega}\left(\sum_{u\in \mathcal{S}}|\mathcal{N}(u)|f(u)\bm{\delta}_{u}\right),\nonumber
\end{align}
and the parameter $\mu$ is chosen as
$$
\mu=\frac{2}{A+B}=\frac{1}{1+\gamma^2}.
$$

The property of frame \cite{Christensen_an_2002} shows that the iteration satisfies
$$
\|{\mathbf f}^{(k)}-{\mathbf f}\|\le \left(\frac{B-A}{B+A}\right)^{k}\|{\mathbf f}^{(0)}-{\mathbf f}\|=\left(\frac{2\gamma}{1+\gamma^2}\right)^{k}\|{\mathbf f}^{(0)}-{\mathbf f}\|.
$$
Then Proposition \ref{pro2} is proved.
\end{proof}

\subsection{Proof of Proposition \ref{cor1}}\label{proof2}
\begin{proof}
According to the definition of local propagation and Table \ref{algIPR}, the iteration of IPR can be written as
$$
{\mathbf f}^{(k+1)}={\mathbf f}^{(k)}+{\bf G}({\mathbf f}-{\mathbf f}^{(k)}),
$$
which is initialized by ${\mathbf f}^{(0)}={\bf G}{\mathbf f}$. Notice that ${\mathbf f}\in PW_{\omega}(\mathcal{G})$ and ${\mathbf f}^{(k)}\in PW_{\omega}(\mathcal{G})$ for any $k$, then ${\mathbf f}^{(k)}-{\mathbf f}\in PW_{\omega}(\mathcal{G})$.
As a consequence of Lemma \ref{lemma1},
$$
\|{\mathbf f}^{(k+1)}-{\mathbf f}\|=\|({\mathbf f}^{(k)}-{\mathbf f})-{\bf G}({\mathbf f}^{(k)}-{\mathbf f})\|\le \gamma\|{\mathbf f}^{(k)}-{\mathbf f}\|,
$$
Proposition \ref{cor1} is proved.
\end{proof}


\footnotesize
\begin{thebibliography}{1}
\bibitem{shuman_emerging_2013}
D. I. Shuman, S. K. Narang, P. Frossard, A. Ortega, and P. Vandergheynst, ``The
  emerging field of signal processing on graphs: Extending high-dimensional
  data analysis to networks and other irregular domains,'' \emph{IEEE Signal
  Process. Mag.}, vol.~30, no.~3, pp. 83-98, 2013.

\bibitem{sandryhaila_discrete_2013}
A. Sandryhaila, and J. M. F. Moura, ``Discrete signal processing on graphs,'' \emph{IEEE Trans. Signal Process.}, vol. 61, no. 7, pp. 1644-1656, 2013.

\bibitem{zhu_graph_2012}
X. Zhu and M. Rabbat, ``Graph spectral compressed sensing for sensor networks,'' in \emph{Proc. 37th IEEE Int. Conf. Acoust., Speech, Signal Process. (ICASSP)}, 2012, pp. 2865-2868.

\bibitem{narang_graph_2012}
S. K. Narang, Y. H. Chao, and A. Ortega, ``Graph-wavelet filterbanks for edge-aware image processing,'' in \emph{ Proc. IEEE Stat.
Signal Process. Workshop (SSP'12)}, 2012, pp. 141-144.

\bibitem{gadde_active_2014}
A. Gadde, A. Anis, and A. Ortega, ``Active semi-supervised learning using sampling theory for graph signals,'' in \emph{Proc. 20th ACM SIGKDD Int. Conf. Knowledge Discovery and Data Mining (KDD'14)}, 2014, pp. 492-501.

\bibitem{narang_signal_2013}
S. K. Narang, A. Gadde, and A. Ortega, ``Signal processing techniques for interpolation in graph structured data,''
in \emph{Proc. 38th IEEE Int. Conf. Acoust., Speech, Signal Process. (ICASSP)}, 2013, pp. 5445-5449.

\bibitem{zhang_graph_2008}
F. Zhang and E. R. Hancock, ``Graph spectral image smoothing using the heat kernel,'' \emph{Pattern Recognition}, vol. 41, no. 11, pp. 3328-3342, 2008.

\bibitem{chen_adaptive_2013}
S. Chen, A. Sandryhaila, J. M. F. Moura, and J. Kovacevic, ``Adaptive graph filtering: Multiresolution classification on graphs,'' in \emph{Proc. 1st IEEE Global Conf. Signal and Inform. Process. (GlobalSIP)}, 2013, pp. 427-430.

\bibitem{Crovella_Graph_2003}
M. Crovella and E. Kolaczyk, ``Graph wavelets for spatial traffic analysis,'' in \emph{Proc. 22nd Annu. IEEE Int. Conf. Comput. Commun. (INFOCOM'03)}, 2003, vol. 3, pp. 1848-1857.

\bibitem{Coifman_Diffusion_2006}
R. R. Coifman and M. Maggioni, ``Diffusion wavelets,'' \emph{Appl. Comput. Harmonic Anal.}, vol. 21, no. 1, pp. 53-94, 2006.

\bibitem{hammond_wavelets_2011}
D. K. Hammond, P. Vandergheynst, and R. Gribonval, ``Wavelets on graphs via spectral graph theory,''
\emph{Appl. Comput. Harmonic Anal.}, vol. 30, no. 2, pp. 129-150, 2011.

\bibitem{narang_perfect_2012}
S. K. Narang and A. Ortega, ``Perfect reconstruction two-channel wavelet filter-banks for graph structured data,'' \emph{IEEE Trans. Signal Process.}, vol. 60, no. 6, pp. 2786-2799, 2012.

\bibitem{agaskar_aspectral_2013}
A. Agaskar, and Y. M. Lu, ``A spectral graph uncertainty principle,'' \emph{IEEE Trans. Inform. Theory}, vol. 59, no. 7, pp. 4338-4356, 2013.

\bibitem{shuman_aframework_2013}
D. I. Shuman, M. J. Faraji, and P. Vandergheynst, ``A framework for multiscale transforms on graphs,'' \emph{arXiv preprint arXiv:1308.4942}, 2013.

\bibitem{ekambaram_multiresolution_2013}
V. N. Ekambaram, G. C. Fanti, B. Ayazifar, and K. Ramchandran, ``Multiresolution graph signal processing via circulant structures,'' in \emph{Proc. IEEE Digital Signal Process., Signal Process. Educ. Meeting (DSP/SPE)}, 2013, pp. 112-117.

\bibitem{zhu_approximating_2012}
X. Zhu and M. Rabbat, ``Approximating signals supported on graphs,'' in \emph{Proc. 37th IEEE Int. Conf. Acoust., Speech, Signal Process. (ICASSP)}, 2012, pp. 3921-3924.

\bibitem{narang_localized_2013}
S. K. Narang, A. Gadde, E. Sanou, and A. Ortega, ``Localized iterative methods for interpolation in graph structured data,''
in \emph{Proc. 1st IEEE Global Conf. Signal and Inform. Process. (GlobalSIP)}, 2013, pp. 491-494.

\bibitem{anis_towards_2014}
A. Anis, A. Gadde, and A. Ortega, ``Towards a sampling theorem for signals on arbitrary graphs,'' in
\emph{Proc. 39th IEEE Int. Conf. Acoust., Speech, Signal Process. (ICASSP)}, 2014, pp. 3892-3896.

\bibitem{thanou_parametric_2013}
D. Thanou, D. I. Shuman, and P. Frossard, ``Parametric dictionary learning for graph signals,'' in \emph{Proc. 1st IEEE Global Conf. Signal and Inform. Process. (GlobalSIP)}, 2013, pp. 487-490.

\bibitem{dong_learning_2014}
X. Dong, D. Thanou, P. Frossard P, and P. Vandergheynst, ``Learning graphs from signal observations under smoothness prior,'' \emph{arXiv preprint arXiv:1406.7842}, 2014.

\bibitem{liu_coarsening_2014}
P. Liu, X. Wang and Y. Gu, ``Coarsening graph signal with spectral invariance,'' in \emph{Proc. 39th IEEE Int. Conf. Acoust., Speech, Signal Process. (ICASSP)}, 2014, pp. 1075-1079.

\bibitem{chen_signal_2014}
S. Chen, A. Sandryhaila, et al. ``Signal inpainting on graphs via total variation minimization,'' in \emph{Proc. 39th IEEE Int. Conf. Acoust., Speech, Signal Process. (ICASSP)}, 2014, pp. 8267-8271.

\bibitem{pesenson_sampling_2008}
I. Pesenson, ``Sampling in Paley-Wiener spaces on combinatorial graphs,''
\emph{Trans. Amer. Math. Soc.}, vol. 360, no. 10, pp. 5603-5627, 2008.

\bibitem{pesenson_variational_2009}
I. Pesenson, ``Variational splines and Paley-Wiener spaces on combinatorial graphs,''
\emph{Constructive Approximation}, vol. 29, pp. 1-21, 2009.

\bibitem{pesenson_sampling_2010}
I. Z. Pesenson, and M. Z. Pesenson, ``Sampling, filtering and sparse approximations on combinatorial graphs,''
\emph{J. Fourier Anal. and Applicat.}, vol. 16, no. 6, pp. 921-942, 2010.

\bibitem{fuhr_poincare_2013}
H. F\"{u}hr and I. Z. Pesenson, ``Poincar\'{e} and Plancherel-Polya inequalities in harmonic analysis on weighted combinatorial graphs,'' \emph{SIAM J. Discrete Math.}, vol. 27, no. 4, pp. 2007-2028, 2013.

\bibitem{shuman_vertex_2013}
D. I. Shuman, B. Ricaud, and P. Vandergheynst, ``Vertex-frequency analysis on graphs,'' no. EPFL-ARTICLE-187669, Elsevier, 2013.

\bibitem{shuman_spectrum_2013}
D. I. Shuman, C. Wiesmeyr, N. Holighaus, and P. Vandergheynst, ``Spectrum-adapted tight graph wavelet and vertex-frequency frames,''
no. EPFL-ARTICLE-190280, Inst. Elect. and Electron. Eng., 2013.

\bibitem{chung_spectral_1997}
F. R. K. Chung, \emph{Spectral Graph Theory}, Amer. Math. Soc., 1997.

\bibitem{Christensen_an_2002}
O. Christensen, \emph{An Introduction to Frames and Riesz Bases}, Springer, 2003.

\bibitem{wang_distributed_2014}
X. Wang, M. Wang, and Y. Gu, ``A distributed tracking algorithm for reconstruction of graph signals,'' to appear in \emph{IEEE J. Selected Topics Signal Process.}, June 2015, available at \emph{arXiv preprint arXiv:1502.0297}.

\bibitem{feichtinger_theory_1994}
H. G. Feichtinger, and K. Gr\"{o}chenig, ``Theory and practice of irregular sampling,''
\emph{Wavelets: Math. and Applicat.}, pp. 305-363, 1994.

\bibitem{grochenig_adiscrete_1993}
K. Gr\"{o}chenig, ``A discrete theory of irregular sampling,''
\emph{Linear Algebra and Its Applicat.}, vol. 193, pp. 129-150, 1993.

\bibitem{marvasti_nonuniform_2001}
F. Marvasti, \emph{Nonuniform Sampling: Theory and Practice}, Springer, 2001.

\bibitem{sauer_iterative_1987}
K. D. Sauer, J. P. Allebach, ``Iterative reconstruction of bandlimited images from nonuniformly spaced samples,''
\emph{IEEE Trans. Circuits and Syst.}, vol. 34, no. 12, pp. 1497-1506, 1987.

\bibitem{marvasti_recovery_1991}
F. Marvasti, M. Analoui, and M. Gamshadzahi, ``Recovery of signals from nonuniform samples using iterative methods,''
\emph{IEEE Trans. Signal Process.}, vol. 39, no. 4, pp. 872-878, 1991.

\bibitem{grochenig_reconstruction_1992}
K. Gr\"{o}chenig, ``Reconstruction algorithms in irregular sampling,''
\emph{Math. Computation}, vol. 59, no. 199, pp. 181-194, 1992.

\bibitem{benedetto_irregular_1992}
J. J. Benedetto,``Irregular sampling and frames,''
\emph{Wavelets: A Tutorial in Theory and Applications}, vol. 2, pp. 445-507, 1992.

\bibitem{pesenson_poincare_2004}
I. Pesenson, ``Poincar\'{e}-type inequalities and reconstruction of Paley-Wiener functions on manifolds,''
\emph{J. Geometric Anal.}, vol. 14, no. 1, pp. 101-121, 2004.

\bibitem{feichtinger_recovery_2004}
H. Feichtinger and I. Pesenson, ``Recovery of band-limited functions on manifolds by an iterative algorithm,''
\emph{Contemporary Math.}, vol. 345, pp. 137-152, 2004.

\bibitem{gleich_matlabbgl}
D. Gleich, The MatlabBGL Matlab Library [Online]. Available: http://www.cs.purdue.edu/homes/dgleich/packages/matlab\_bgl/index.html.

\end{thebibliography}
\end{document}